\documentclass[letterpaper,11pt]{article}
\usepackage{amsthm}  
\usepackage{amsmath, amssymb}
\usepackage{color}
\usepackage{fullpage}
\usepackage[numbers]{natbib}
\usepackage[noend]{algorithmic}
\usepackage{thmtools, thm-restate}
\usepackage{tikz}
\usepackage{enumerate}
\usepackage{thm-restate}
\usepackage{graphicx}
\usepackage{amsopn}
\usepackage{caption}
\usepackage{hyperref}
\usepackage{subcaption}
\usepackage{zerosum,csquotes}
\usepackage{enumitem,linegoal}
\usepackage[linesnumbered,ruled,vlined]{algorithm2e}
\usepackage{comment}
\usepackage{float}
\usepackage{tcolorbox}
\usepackage{wrapfig}

\usepackage{ctable}					

\usepackage{mathtools}
\usepackage[flushleft]{threeparttable}
\usepackage[normalem]{ulem}

\usepackage{footnote}
\makesavenoteenv{tabular}
\makesavenoteenv{table}

\usepackage[margin=1in]{geometry}

 \newtheorem{theorem}{Theorem}[section]
 \newtheorem{lemma}[theorem]{Lemma}
  \newtheorem{claim}[theorem]{Claim}

\theoremstyle{definition}
 \newtheorem{definition}[theorem]{Definition}

\makeatletter
\newif\ifqed
\def\GrabProofArgument[#1]{ #1: \egroup\ignorespaces}
\def\proof{\noindent\textbf\bgroup Proof%
	\@ifnextchar[{\GrabProofArgument}{. \egroup\ignorespaces}\global\qedtrue}

\def\qedhere{\ifmmode\tag*{\qedsign}\else\hspace*{\fill}\qedsign\medskip\fi\global\qedfalse}
\def\qedsign{$\Box$}
\makeatother

\newcommand{\aliz}[1]
{\par {\color{blue} Alireza: #1 \par}}

\definecolor{mygreen}{RGB}{20,140,80}
\definecolor{mylightgray}{RGB}{230,230,230}

\definecolor{mygreen}{RGB}{20,140,80}
\definecolor{mydarkgray}{gray}{0.15} 
\definecolor{oceanblue}{HTML}{2c55c2}
\hypersetup{
     colorlinks=true,
     citecolor= mygreen,
     linkcolor= black
}

\SetCommentSty{mycommfont}

\newcommand{\mcal}[1]{\ensuremath{\mathcal {#1}}}

\newcommand{\algA}{{\ensuremath{\mcal{A}}}\xspace}

\newcommand{\poly}{\ensuremath{{{\sf poly}}}\xspace}

\newcommand{\Alg}{\ensuremath{{\sf Alg}}\xspace}

\newcommand{\tr}{\ensuremath{{\sf tr}}\xspace}
\newcommand{\prefix}{\ensuremath{{\sf pre}}\xspace}

\newcommand{\Expt}{\ensuremath{{\sf Expt}}\xspace}

\newcommand{\Geom}{\ensuremath{{\sf Geom}}\xspace}

\newcommand{\ds}{\ensuremath{d}\xspace}

\newcommand{\err}{\ensuremath{{\sf err}}\xspace}
\newcommand{\dgr}{\ensuremath{{\sf deg}}\xspace}
\renewcommand{\deg}{\ensuremath{{\sf deg}}\xspace}

\newcommand{\PSum}{\ensuremath{{\sf PSum}}\xspace}
\newcommand{\Cnt}{\ensuremath{{\sf Cnt}}\xspace}
\newcommand{\thresh}{\ensuremath{{\sf T}}\xspace}
\newcommand{\noise}{\ensuremath{{\mcal{E}}}\xspace}
\newcommand{\noisep}{\ensuremath{{\mcal{N}}}\xspace}

\newcommand{\lap}{Lap}

\newcommand{\opt}{OPT}

\newcommand{\idx}{\ensuremath{\mathit{idx}}\xspace}

\newcommand{\elaine}[1]{{\footnotesize\color{magenta}[Elaine: #1]}}
\renewcommand{\elaine}[1]{}

\newtheorem{fact}{Fact}

\newcommand*\samethanks[1][\value{footnote}]{\footnotemark[#1]}

\newcounter{proccnt}

\newcommand{\konote}[1]{}

\title{Differentially Private Densest Subgraph}
\author{
	Alireza Farhadi\thanks{University of Maryland. Email: \texttt{\{farhadi,hajiagha\}@cs.umd.edu}.}
	\and
	MohammadTaghi Hajiaghayi\samethanks[1]
	\and
	Elaine Shi\thanks{Carnegie Mellon University. Email: \texttt{runting@cs.cmu.edu}.} \thanks{Supported by a grant from ONR, a gift from Cisco, and NSF awards under grant numbers 2128519 and 2044679, and a Packard Fellowship.}
}

\begin{document}
	\newcommand{\ignore}[1]{}
\renewcommand{\theenumi}{(\roman{enumi})}
\renewcommand{\labelenumi}{\theenumi.}
\sloppy

%
%

\date{}

\maketitle


\begin{abstract}
Given a graph, the densest subgraph problem
asks for a set of vertices such that the average degree 
among these vertices is maximized.
Densest subgraph has numerous applications in learning, e.g., community detection
in social networks, link spam detection, correlation mining, bioinformatics, 
and so on. 
Although there are efficient algorithms that output either
exact or approximate solutions  
to the densest subgraph problem,
existing algorithms may violate the privacy of the individuals in the network,
e.g., leaking the existence/non-existence of edges.

In this paper, we study the densest subgraph problem in the framework of the differential privacy, and we derive upper and lower bounds for this problem. We show that there exists a linear-time $\epsilon$-differentially private algorithm that finds a 2-approximation of the densest subgraph with an extra poly-logarithmic additive error. 
Our algorithm not only reports the approximate density of the densest subgraph, but also reports the vertices that form the dense subgraph.

Our upper bound almost matches the famous 
$2$-approximation by Charikar both in performance  
and in approximation ratio, but we additionally achieve differential privacy. 
In comparison with Charikar's algorithm, our algorithm has an extra
poly-logarithmic additive error. 
We partly justify the additive error 
with a new lower bound, 
showing that  
for any differentially private algorithm 
that provides a constant-factor approximation, a  
sub-logarithmic additive error is inherent. 

We also practically study our differentially private algorithm on real-world graphs, and we show that in practice the algorithm finds a solution which is very close to the optimal.


\end{abstract}

\section{Introduction}
\label{sec:intro}

The {\it densest subgraph problem} (DSP)~\cite{goldbergdensestsubgraph} 
is a fundamental tool to many graph mining applications. 
Given an undirected graph $G=(V,E)$, the {\it density} of an induced subgraph $S \subseteq V$ 
is defined as $\ds_G(S)= |E(S)|/ |S|$, where $E(S)$ is the set of all edges in the subgraph induced by the vertices $S \subseteq V$. In the densest subgraph problem, the goal is to find a subset of vertices $S \subseteq V$ with the highest density $\ds_G(S)$.
The densest subgraph problem (DSP) is used as a crucial tool for community detection in social network graphs. This problem also has notable applications in learning including 
link spam detection~\cite{dspspam}, 
correlation mining~\cite{dspcorrelation}, story identification~\cite{dspstory} and bioinformatics~\cite{dspbio}
-- we refer the reader to the tutorial by Gionis and Tsourakakis~\cite{dsdtutorial} 
for more applications of DSP. 
Due to its importance, the DSP problem has been studied extensively 
in the literature~\cite{goldbergdensestsubgraph,charikardensestsubgraph,khullerdensestsubgraph,dsp00,dsp01,dsp02}: \elaine{cite more}
it is long known that 
efficient, polynomial-time algorithms exist 
for finding the exact solution of DSP~\cite{goldbergdensestsubgraph,charikardensestsubgraph,khullerdensestsubgraph}.

In many applications of DSP, however, the underlying graph is privacy sensitive (e.g.,
social network graphs).  
Therefore, one might be concerned that the result output by the DSP  
algorithm might breach the privacy  
privacy of the individuals in the network, e.g., disclose the (non)-existence
of friendship
between pairs of individuals.
In this paper, we explore how to 
perform community detection on sensitive graphs, while protecting
individuals' privacy.
To this end, we ask the following question,  

\begin{tcolorbox} 
{\it Can we construct a differentially private
algorithm that computes a good approximation of the densest subgraph of a given
graph $G$?}
\end{tcolorbox}

We first help the reader recall the notion of differential 
privacy~\cite{calibrating} in a graph context.
Let $G$ and $G'$ be two graphs that are identical expect the existence/non-existence of a single edge. Informally, the differential privacy requires that the outputs of the (randomized) 
algorithm on $G$ and $G'$ are close in distribution.
In this way, 
the output of the algorithm does not reveal meaningful information about the existence of an edge in the graph. Henceforth, 
we say that two undirected graphs $G$ and $G'$ are \textit{neighboring} if they differ 
in only one edge.
More formally, differential privacy (DP) is defined as follows~\cite{calibrating}.

\begin{definition}[$(\epsilon, \delta)$-Differential Privacy (DP)]
Let $\epsilon > 0$ and $\delta \in [0, 1]$.
We say that a randomized algorithm  
\Alg achieves $(\epsilon, \delta)$-differential privacy or $(\epsilon, \delta)$-DP
for short, 
if for any two {\it neighboring}
graphs $G$ and $G'$, for any subset $U$ of the output space, 
\[
\Pr[\Alg(G) \in U] \leq 
e^\epsilon \cdot \Pr[\Alg(G') \in U]  + \delta
\]
Whenever $\delta = 0$, we also say that the algorithm
satisfies $\epsilon$-DP.
\end{definition} 

For the densest subgraph problem, 
we assume that the output contains 
1) a dense subset of vertices $S \subseteq V$; and 2)
an estimate of the density of $S$.
Had we required the algorithm to output \textit{only} an estimate
of $\max_{S \subseteq V}\ds_G(S)$, i.e., an estimated density of the densest
subgraph, 
then it would have been easy to devise a DP algorithm:  
observe that the quantity $\max_{S \subseteq V}\ds_G(S)$
has small global sensitivity, that is, if we flip
the existence of a single edge in 
$G$, the quantity $\max_{S \subseteq V}\ds_G(S)$ 
changes by at most $1$.  As a result, we can just use
the standard Laplacian mechanism~\cite{calibrating} to output 
a DP estimate with good accuracy.
However, 
we stress that most interesting applications would also
want to know the dense community $S$ --- simply knowing
an estimate of its density would not be too useful. 

The requirement to also report a dense vertex set 
$S \subseteq V$ makes it much more challenging to devise a DP algorithm.
In our case, none of the off-the-shelf DP mechanisms would 
directly work to the best of our knowledge.
First, observe that the output is high-dimensional, and 
has high global sensitivity as we explain in 
Appendix~\ref{sec:naive}. 
Therefore, the standard Laplacian mechanism~\cite{calibrating,dpbook-salil,dpbook-dr} 
(also called output perturbation)
completely fails. 
Another na\"ive approach is randomized response~\cite{randomizedresponse,dpbook-salil,dpbook-dr} 
(also called input perturbation),
i.e.,
adding some noise to obfuscate the existence of each edge.
Unfortunately, 
as we argue in Appendix~\ref{sec:naive}, the randomized
response approach gives poor utility. 
Finally, the exponential mechanism~\cite{expmech,dpbook-salil,dpbook-dr} 
also fails --- not only is it not polynomial-time, the standard
analysis gives an error bound  
as large as $O(n)$ which makes the result meaningless.

\ignore{
 We study the densest subgraph problem under this notion of privacy, and we derive an upper bound and a lower bound for the accuracy of an algorithm that finds DSP while preserving the differential privacy. Considering any two neighboring graphs, it is easy to see that the density of the densest subgraph is different for these two graphs by at most $1$. Denoting by $\Delta$, we define the \textit{global sensitivity} of DSP to be the maximum difference between to neighboring graphs. Therefore, the global sensitivity of DSP is at most 1. It is known that for a function $f$ with the global sensitivity of $\Delta_f$, we can make the output of the function $\epsilon$-DP by adding a noise sampled from $\lap(\Delta_f/\epsilon)$. Since the global sensitivity of DSP is at most 1, we can directly imply that we can report the \textit{density} of DSP by having $O(1/\epsilon)$ additive error. 
 
Although the \textit{density} of DSP has a low sensitivity, the vertices that form the densest subgraph can change significantly if we add or remove one edge. For example, consider a graph $G$ with $n$ vertices which is the union of two disjoint cliques $G_1, G_2$ of size $n/2$. \aliz{figure here?} A deterministic for DSP algorithm should report one of these two cliques as the densest subgraph. Let assume the algorithm reports $G_1$ as the densest subgraph. If we remove any edge from $G_1$, the unique densest subgraph of the new graph is $G_2$. Therefore, all vertices in the DSP can change by adding or removing only one edge. Roughly speaking, this example shows that the \textit{vertex set} of DSP has a high sensitivity w.r.t. adding or removing edges. As a consequence, reporting the vertex set of DSP is more challenging in the framework of differential privacy.

In this paper, we consider the problem of reporting the \textit{vertex set} of DSP in addition to its density, and we give the first upper-bound and lower-bound for this problem. 
}

\subsection{Our Results and Contributions}
We present new upper- and lower-bounds for the differentially private,
densest subgraph problem. 
First, 
we give a linear-time DP approximation algorithm for 
the densest subgraph problem. 
The runtime and accuracy of our algorithm are roughly competitive to  
the state-of-the-art {\it non-private} approximation algorithm
by Charikar~\cite{charikardensestsubgraph}. 
Specifically, let $n$ denote the number of vertices. 
Our algorithm is \textit{linear-time}, and achieves 
$\epsilon$-DP and 
$(2, O(\frac{1}{\epsilon} \log^{2.5}n))$-approximation --- here
we have two approximation parameters: the first parameter $2$
is the {\it multiplicative} approximation ratio, and the second parameter
$O(\frac{1}{\epsilon} \log^{2.5}n)$
is an {\it additive} error.
In comparison,  
Charikar's famous (non-private) linear-time algorithm 
achieves $(2, 0)$-approximation where the additive error is $0$.
We justify the extra additive error
with a new lower bound, showing that to achieve
any constant-multiplicative approximation, 
some sub-logarithmic additive error is unavoidable.
Our upper- and lower-bound
results are stated in the following theorems:

\begin{theorem}[DP approximation of densest subgraph]
Given a graph $G$, and parameters $\epsilon>0$ and $\sigma \in [0,1]$, there exists a linear-time $\epsilon$-DP algorithm that succeeds with the probability of $1-\sigma$ and 
outputs $S \subseteq V$ and an estimate $\ds^*$ such that
$$ \opt/2 - O(\frac{1}{\epsilon} \cdot \log^{2.5}n \cdot \log\frac{1}{\sigma}) \le \ds_G(S) \le \opt, \ \ \text{ and }  \ \
|\ds^* - \ds_G(S)| \leq O(\frac{1}{\epsilon} \cdot \log^{2.5}n \cdot \log\frac{1}{\sigma})
$$
where $\opt$ is the true density of the densest subgraph.
\end{theorem}

\begin{theorem}[Lower bound on additive error for DP densest subgraph]
\label{thm:lb}
Let $\alpha > 1$, $\epsilon > 0$ be arbitrary constants,
$\exp(-n^{0.49}) < \sigma < 0.000001 \cdot \min(1, \epsilon, \exp(-\epsilon))$, and $0 \le \delta \le \frac{\sigma\epsilon}{\log \frac{1}{4\sigma}}$.
Then, there exists a sufficiently small 
$\beta = \Theta\big(\frac{1}{\alpha}\sqrt{\frac{1}{\epsilon}\log \frac{1}{\sigma}}\big)$ 
such that 
there does not exist an $(\epsilon,\delta)$-DP mechanism that 
achieves $(\alpha, \beta)$-approximation 
with $1-\sigma$ probability.
\end{theorem}

Note that our upper bound achieves $\epsilon$-DP, and our lower 
bound works even for $(\epsilon, \delta)$-DP. This makes both our 
upper- and lower-bounds stronger. 
The proof of Theorem \ref{thm:lb} is available in Appendix \ref{sec:lb}. Finally, we conclude the paper in Section \ref{sec:exp} by demonstrating the performance of our algorithm on real-world datasets. We show that in practice, our algorithm achieves a very accurate solution, even for small choices of the privacy parameter $\epsilon$. 

\subsection{Technical Highlight}
\label{sec:roadmap}
To see the intuition behind our final algorithm, it helps
to break it down into several intermediate steps, to see how the
various techniques are eventually woven together.

\paragraph{Background on Charikar's famous algorithm.}
Our algorithm is inspired by a work of Charikar~\cite{charikardensestsubgraph}. Charikar~\cite{charikardensestsubgraph} 
shows that a simple greedy algorithm can achieve a multiplicative 
approximation ratio of $2$ for DSP. The greedy algorithm is as follows. Let $G=(V,E)$ be an undirected graph. Initially, let $S :=V$, i.e., $S$ is initalized to 
the set of all vertices. At each iteration, the algorithm finds a vertex $v_{\min} \in S$ with the minimum degree in the graph induced by the vertices of $S$, and removes $v_{\min}$ from $S$. Consider an algorithm that repeats the aforementioned procedure until the set $S$ becomes empty. From all of the sets $S$ encountered during the execution of the algorithm, 
the algorithm returns the one with the highest density. Charikar proved  
that this simple greedy algorithm achieves an approximation ratio of $2$.

\paragraph{Warmup idea: a quadratic-time DP algorithm.}
Our first idea is as follows.
In Charikar's algorithm, in each iteration, 
all residual vertices $v \in S$ examine their degree
within the subgraph induced by $S$ --- henceforth 
we call the 
the degree of $v$ in the subgraph induced by $S$ 
the {\it residual degree} of $v$.
Charikar's algorithm picks the $v$ with the minimum 
residual degree
and removes it from $S$. 
Our idea is to replace the residual degree with a noisy, DP counterpart. 
Unfortunately, na\"ively adding independent noise 
to the true residual degree in each of the $n$ iterations 
would result in an $n$-fold loss in error  
given a fixed privacy budget $\epsilon$ (and the
loss can be reduced to $\sqrt{n}$ if we  
allowed $(\epsilon, \delta)$-DP rather than $\epsilon$-DP
and used the advanced composition theorem~\cite{boostingdp,dpbook-salil,dpbook-dr}).


Our idea is to rely on the elegant DP prefix sum mechanism
by Dwork et al.~\cite{continual} and Chan, Shi, and Song~\cite{icalp10,tissec11}. 
Specifically, we can think of the problem as follows.
\begin{itemize}[leftmargin=5mm]
\item
Initially, every vertex computes its noisy total degree
using the standard Laplacian mechanism. Although there are $n$ vertices,
we only need to add noise of constant average magnitude by using parallel
composition.

\item 
Next, every vertex $u$ still remaining in $S$ maintains a noisy counter 
to keep track of 
roughly how many of its direct neighbors have departed (i.e., 
have been removed from $S$).
If we subtract this value from the vertice's noisy total degree,
we get an estimate of its residual 
degree in the subgraph induced by $S$.
\end{itemize}
Therefore, 
the problem boils down to how to have every residual vertex
$v$ maintain a noisy counter of how many of its neighbors have departed. 
Imagine that every time a neighbor of $v$ departs, a value of $1$
is accumulated to $v$'s counter; and 
every time a non-neighbor of $u$ departs, 
$0$ is accumulated to $v$'s counter.
Using the 
elegant DP prefix sum mechanism
by Dwork et al.~\cite{continual} and Chan et al.~\cite{icalp10,tissec11},
we can report $v$'s noisy counter value
at any time step, incurring 
only $O(\frac{1}{\epsilon} \cdot \poly\log n)$ 
error with all but negligible probability.
The noisy counter values and the vertices' noisy degrees are then used 
to determine which vertex is to depart next.
Further, although it seems like there are $n$ counters, using
parallel composition, we need not incur extra loss 
in the privacy budget
due to the $n$ counters.

By extending Charikar's proof (which we omit in this short roadmap), 
we can prove
that this warmup algorithm
achieves the desired $(2, \frac{1}{\epsilon} \cdot \poly\log n)$-approximation.
In particular, 
the error of the prefix sum 
mechanism directly contributes to the additive error term.
Unfortunately, the warmup algorithm incurs 
$\Theta(n^2)$ runtime, since we need to update $O(n)$ noisy counters
in each of the $n$ iterations.

\paragraph{Making it quasilinear time.}
Our final goal is to get an $O(m + n)$-time algorithm where $m$ denotes
the number of edges and $n$ denotes the number of vertices.
However, as an important stepping
stone, let us first consider how to make it
{\it quasilinear} time in $m + n$.
The key observation is that when the graph is sparse, 
updates to the vertices' noisy counters
(realized by the prefix sum mechanisms)
are also sparse. 
Most of the $n^2$ updates come with an input $0$, and only $m$ of them  
come with an input of $1$.
Our idea is therefore to avoid triggering
updates when the input is $0$. 

While the intuition seems simple, realizing this idea differentially privately 
is actually tricky since 
we need to avoid consuming too much privacy budget. 
At a very high level, 
each vertex $v \in S$ keeps track of a noisy outstanding counter ${\sf Cnt}(v)$
of the number of its neighbors 
that departed recently, but 
have not been accumulated into the prefix sum mechanism yet.
When a vertex $u$ gets removed from $S$, 
it informs its neighboring vertices to 
update their noisy outstanding counters.
At this moment, each  vertex $v$ also checks
if its noisy outstanding 
counter ${\sf Cnt}(v)$
has exceeded some predetermined polylogarithmic noisy  
threshold --- if so, it accumulates  
the current outstanding counter 
into its prefix sum mechanism, and resets 
${\sf Cnt}(v)$ to $0$.

The key technical challenge here is that we would be invoking with high probability the total of $O(m)$ updates
 to the vertices' 
noisy outstanding counters, but we cannot afford an $O(m)$-fold loss
in the privacy budget (or equivalently, an $O(m)$-fold loss in error 
when the privacy budget is fixed).
To resolve this problem, 
our idea is {\it in spirit}
reminiscent of the {\it sparse-vector technique}~\cite{complexitydp,medianmech,dpmw}.
We show that we can reduce the privacy analysis of our algorithm to the standard 
sparse-vector technique.
See the subsequent  ``{\bf proof techniques}''
paragraph regarding the technical challenges in the analysis and proof.

\ignore{
The idea is that we do not want to invoke
updates to each vertex's noisy counter (realized by the prefix sum mechanism)
too often.
Updates to a vertex's prefix sum mechanism should only be triggered
when a non-zero 
Our idea is inspired by the sparse-vector technique~\cite{}, \elaine{FILL} 
but our actual instantiation and analysis differ from the standard
description of the sparse-vector technique.

While the intuition seems simple, to make the algorithm and analysis work
is non-trivial. The key challenge here 
is that each vertex $v$ needs to 
compute its outstanding counter
value ${\sf Cnt}(v)$ for up to $n$ steps.
}

\paragraph{Final touches: making it linear time.}
The above algorithm 
can be implemented in $O(m + n) \log n$ time 
if we use a Fibonacci heap to store the residual vertices 
based on their residual degree (i.e., degree in the graph induced by $S$).
To make the algorithm linear time, 
we discretize 
vertices residual degree into polylogarithmically sized regions,
and place each 
vertice in a corresponding bucket based on its residual degree.

Using an idea inspired 
by Charikar~\cite{charikardensestsubgraph}, 
one can show that 
if a vertex is removed from the $k$-th bucket 
in the current iteration, then, in the next iteration, we only need
to sequentially look at the $(k-1)$-th, $k$-th, and $(k+1)$-th, ... buckets.
Moreover, within each bucket, all vertices are treated as having
roughly the same degree, and 
we do not further differentiate them in picking the next vertex to remove from $S$.
Using appropriate data structures
to store the buckets and the vertices within the buckets, 
we can eventually obtain a DP-algorithm
that completes in $O(m + n)$ runtime. 
Here, the discretization due to bucketing introduces some additive 
polylogarithmic error, but asymptotically 
we still preserve the $(2, O(\frac1\epsilon \log^{2.5}n))$-approximation as before.
We defer the detailed 
algorithm and analysis to Section~\ref{sec:linear}. 

\paragraph{Proof techniques.}
Proving our algorithm DP turns out to be rather non-trivial. Specifically, our
algorithm is not a simple sequential
composition of the various underlying building blocks (e.g., prefix sum  
mechanism, and outstanding counter threshold queries). 
Therefore, we 
cannot simply analyze each building block separately and then use
standard composition theorems
to get the desired DP guarantees.
The problem is that the building blocks are interleaved in an adaptive way:
the outcome of one step of the prefix sum mechanism
corresponding to some vertex
will affect 
the input to the  
next step of some outstanding counter query, which will then affect
the input to the next step of prefix sum mechanism. 
Despite the complex and adaptive nature of our algorithm, we show that the privacy analysis of our algorithm can be reduced the privacy bounds of sparse-vector-technique.
The actual proof is involved and we defer
the detailed exposition to Section~\ref{sec:dpproof}, 
Appendix~\ref{sec:utilproof}, and Appendix~\ref{sec:linearproof}. 
All the missing proofs are available in the appendices with the same Theorem number.
\elaine{say more?}

\ignore{
\subsection{Old Text}

Our differential private algorithm is inspired by this simple greedy algorithm. We use prefix sum mechanism introduced by Dwork et al. \ref{???} in designing our algorithm, and we show that using the prefix sum mechanism we can efficiently find a vertex $i^*_{\min}$ that has a small degree. Although, the vertex $i^*_{\min}$ found by our algorithm might not have the minimum degree, we show that the degree of this vertex is at most $d_S(i_{\min}) + \log^{2.5} n$ where $d_S(i_{\min})$ is the minimum degree of a vertex in the induced subgraph of $S$. In section \ref{???}, we use this tool, and we achieve an algorithm with the multiplicative approximation factor $2$ that has an additive error of $O(\log^{2.5} n)$. In section \ref{???} we improve the running time of our algorithm and we show that it can be implemented in the linear time.

\begin{theorem}
Given a graph $G$, and parameters $\epsilon>0$ and $\delta, \sigma \in [0,1]$, there exists a $(\epsilon, \delta)$-DP that succeeds with the probability of $1-\sigma$ and finds a subgraph $S$ such that
$$ opt/2 - O(\frac{1}{\epsilon} \cdot \log^{2.5}n \cdot \log\frac1\sigma) \le \ds_G(S) \le opt \,,$$
where $opt$ is the density of the densest subgraph.
\end{theorem} 
}

\subsection{Additional Related Work} 

\paragraph{Differentially private algorithms for graphs.}
Early works on differentially private graph algorithms
focused on computing simple statistics 
from graphs.
The elegant work by Nissim et al.~\cite{smoothsens}
was the first to apply the DP notion to graph computations. 
Specifically, they showed how to release the cost of minimum spanning tree and the number of
triangles in a graph.
The work by Karwa et al.~\cite{dpgraphkarwa}
extended triangle counting to counting other subgraph structures differentially privately. 
The work by Hay et al.~\cite{dpgraphhay} 
considered how to release degree distribution while preserving DP.
Other works consider how to release the answers to all queries  
belonging ot some class on a given graph.
For example,  Gupta, Roth, and Ullman~\cite{dpidc}
consider how to compute a private synthetic data structure for answering
all cut queries with $O(n^{1.5})$ error  where $n$ denotes the number of vertices.
Gupta, Hardt, Roth, and Ullman show how to 
release the cut function on arbitrary graphs~\cite{dpconj}.

\paragraph{Closely related work.}
Our DSP problem can be viewed as a combinatorial optimization problem.
To the best of our knowledge, there exist few works that consider
how to solve combinatorial optimization problems 
differentially privatey in graphs.
The first such work was the elegant work by  
Gupta et al.~\cite{dpcombopt}.
They showed polynomial-time DP algorithms for 
approximating the min-cut and vertex cover problem. For min-cut, their
solution can report the vertices on both sides of the cut. 
However, for vertex cover, their algorithm 
cannot report the exact set of vertices in the vertex cover --- instead, it outputs
a permutation of vertices, and if one knows the set of edges, one can recover a good vertex cover from this permutation.

\paragraph{Other notions of privacy.}
In this paper, we consider the notion of edge differential privacy 
in graphs, which was a standard notion adopted
in various prior works~\cite{dpcombopt,dpconj,dpidc,smoothsens,dpgraphkarwa,dpgraphhay}. 
Some works study stronger notions.
For example, 
Kasiviswanathan et al.~\cite{nodedp} 
and Blocki et al.~\cite{nodedp-blocki}
investigate the notion of node differential privacy,
where neighboring graphs are defined as two graphs that differ in 
one node rather than one edge. 
Gehrke et al.~\cite{zkpriv-graph}
explore a different strengthening of differential privacy for social network graphs. 
An interesting future work direction is to understand whether we can design
accurate DSP algorithms that satisfy these strengthened notions of privacy.

\paragraph{Concurrent work of Nguyen and Vullikanti \cite{nguyen2021differentially}.} In the independent and concurrent work, Nguyen and Vullikanti designed a $(\epsilon, \delta)$-DP algorithm for the densest subgraph. Their $(\epsilon, \delta)$-DP algorithm is less secure than our $\epsilon$-DP algorithm since for our algorithm we have $\delta=0$. We also show in section \ref{sec:exp} that in practice, our algorithm achieves a significantly more accurate solution  in comparison to Nguyen and Vullikanti. It is also worth mentioning that while our algorithm runs in a \textit{linear-time}, it is not clear whether the algorithm proposed by Nguyen and Vullikanti can be realized in a linear work. In fact, the algorithm $\textsc{SEQDENSEDP}$ introduced in their paper \cite{nguyen2021differentially} has a quadratic running time. Nguyen and Vullikanti also provide a similar sub-logarithmic lower bound on the additive error of any differentially private algorithm. They also consider a PRAM version of their algorithm, while our algorithm is in the RAM model.

\ignore{
1. why this is an important problem

2. why previous approaches don't work

3. our the result is cool in the sense that we can output the vertex set, whereas
most prior DP algorithms for graphs cannot do that.
}

\section{Preliminaries}

\subsection{Densest Subgraph}
We define the densest subgraph problem~\cite{goldbergdensestsubgraph,charikardensestsubgraph}.
Let $G = (V, E)$ be an undirected graph and $S \subseteq V$.
We define $E(S)$ to be the edges induced by $S$, i.e., 
$E(S) := \{(i, j) \in E : i,j \in S\}$.

\begin{definition}
Let $S \subseteq V$. We define the density $d(S)$
of the subset $S$ to be 
$d(S) := \frac{|E(S)|}{|S|}$.
We define the density $d(G)$ of the undirected graph $G(V, E)$ to be 
$d(G) := \max_{S \subseteq V} d(S)$.
\end{definition}

Observe that $2 d(S)$ 
is simply the average degree of the subgraph induced by $S$.

\label{sec:dpds}

We next define 
the notion of approximation we use to measure
the algorithm's utility.

\begin{definition}[Approximation algorithm for densest subgraph]
Given an undirected graph $G = (V, E)$, 
we want to design a randomized algorithm \Alg that outputs 
1) a subset of vertices $S^* \subseteq V$ which is an estimate
of the densest subgraph; and 
2) a noisy density $d^*$, which is an estimate of $d(G)$.
Let $\alpha \geq 1$, $\beta > 0$, and $\sigma \in (0, 1]$. 
Such an algorithm \Alg is said to achieve $(\alpha, \beta)$-approximation 
with $1-\sigma$ probability, 
iff with $1-\sigma$ probability, the following hold:
\begin{enumerate}[leftmargin=5mm,itemsep=1pt]
\item 
$d(S^*) \geq d(G)/\alpha - \beta$, i.e.,
the algorithm outputs a dense set $S^*$ whose (true) density 
is close to $d(G)$; and
\item 
$|d^* - d(S^*)| \leq \beta$, i.e., the estimated density $d^*$
is close to the true density of the reported subgraph $S^*$. 
\end{enumerate}
\ignore{
We define the {\it approximation factor} and {\it error}
of the algorithm as below:
\begin{enumerate}[leftmargin=5mm,itemsep=1pt]
\item 
{\it Approximation factor.}
Let $\alpha \geq 1$ and $\sigma \in [0, 1]$.
We say that the algorithm 
achieves $(\alpha, \sigma)$ approximation iff 
with $1-\sigma$ probability, $d(S^*) \geq d(G)/\alpha$.
In other words, 
with $1-\sigma$ probability, the (true) density of reported set $S^*$ 
must be at least $1/\alpha$
fraction of the density of the densest subgraph of $G$. 

\item 
{\it Error.}
The algorithm's error is characterized  
by the distribution $|d^* - d(S^*)|$, i.e., 
how different the reported noisy density $d^*$ is w.r.t. 
to the true density of the output set $S^*$.
The error is a random variable, and later on we will often 
characterize its magnitude by its expectation or its tail bound.
\end{enumerate}
}
\label{defn:alg}
\end{definition}

As mentioned in Section~\ref{sec:intro}, 
we consider edge differential privacy in this paper. 
We say that two undirected graphs $G$ and $G'$ are {\it neighboring}
if the adjacency matrix of $G$ and $G'$ differ in only one entry (i.e.,
$G$ and $G'$ are the same except for the existence/non-existence
of a single edge).
The notion of $(\epsilon, \delta)$-differential privacy 
and $\epsilon$-differential privacy were formally defined 
in Section~\ref{sec:intro}.


\subsection{Mathematical Tools}
We define the symmetric geometric distribution~\cite{dpfinite,ndss11}
which can be viewed as a discrete version of the standard Laplacian
distrubtion~\cite{calibrating}.

\begin{definition}[Symmetric geometric distribution]
Let $\gamma > 1$.  The symmetric geometric
distribution $\Geom(\gamma)$ takes integer values
such that the probability mass function at $k$ is $\frac{\gamma - 1}{\gamma + 1} \cdot \gamma^{-|k|}$.
\end{definition}
We shall assume that sampling from the symmetric geometric 
distribution takes constant time.
How to sample such noises was discussed in detail in earlier
works on differential privacy~\cite{dpfinite}.

The global sensitivity of a function 
$f({\bf I})$, denoted $\Delta_f$, is defined as follows:
\[\Delta_f := 
\max_{{\bf I}, {\bf I}' \text{neighboring}}|f({\bf I}) - f({\bf I}')|_1\]

The following fact about the geometric mechanism (which is equivalent
to a discrete version of the Laplacian mechanism) was 
shown in previous works~\cite{calibrating,continual,icalp10,tissec11,dpfinite}.
\begin{fact}[Geometric mechanism]
The geometric mechanism 
$f'({\bf I}) := f({\bf I}) + \Geom(\exp(\epsilon/\Delta_f))$ satisfies $\epsilon$-DP. 
Moreover, 
for $\sigma \in (0, 1)$, for any input ${\bf I}$, the error 
$|f'({\bf I})-f({\bf I})|$
is upper bounded by 
$O(\frac{\Delta_f}{\epsilon} \cdot \log\frac1\sigma)$ with probability $1-\sigma$.
\label{fct:geommech}
\end{fact}

\subsection{Building Block: Differentially Private Prefix Sum Mechanism}
Dwork et al.~\cite{continual} as well as Chan, Shi, and 
Song~\cite{icalp10,tissec11}
suggest a DP prefix sum mechanism.  
Initially, the mechanism is initialized
with $N$, which is an upper bound on the total number of  
values that will arrive, and at the time of initialization,
the mechanism's output is defined to be $0$.
Next, a sequence of at most $N$ integer values 
arrive one by one, and the
value that arrives at time $t \in [N]$ is denoted $x_t$.
In every time step $t$, 
the mechanism outputs
an estimate $\PSum_t$ of the prefix sum $\sum_{t' \leq t} x_{t'}$.
The term $\err_t := |\PSum_t - \sum_{t' \leq t} x_{t'}|$
measures the {\it error} of the estimate at time $t$.

We say that two integer sequences ${\bf x} := (x_1, \ldots, x_N)$ and 
${\bf x}' := (x'_1, \ldots, x'_N)$   
are neighboring, if 
the vector ${\bf x} - {\bf x}'$
has exactly one coordinate that is either $1$ or $-1$, 
and all other coordinates are $0$.
A prefix sum mechanism for length-$N$ sequences 
is said to satisfy $\epsilon$-adaptive-DP, iff
for any {\it admissible} (even unbounded) adversary $\algA$,
for any set $\Gamma$, 
\[
\Pr[\Expt^{\algA}_0 \in \Gamma]
\leq 
e^{\epsilon} \cdot \Pr[\Expt^{\algA}_1 \in \Gamma]
\]
where for $b\in \{0, 1\}$, $\Expt^{\algA}_b$ is defined as follows:

\begin{tcolorbox}
\begin{center}
$\Expt^{\algA}_b$:
\end{center}
\begin{itemize}[leftmargin=5mm,itemsep=1pt,topsep=5pt]
\item 
Initialize a prefix sum mechanism denoted \PSum.
\item 
For $t := 1, 2, \ldots, N$:
\begin{itemize}[leftmargin=5mm,itemsep=1pt]
\item 
\algA outputs the next values $x^{(0)}_t$ and $x^{(1)}_t$;
\item 
Input $x^{(b)}_t$ to \PSum, and send $\PSum$'s new output to \algA.
\end{itemize}
\item 
Output $\algA$'s view which includes 
the sequence of all outputs produced by \PSum. 
\end{itemize}

\paragraph{Admissible \algA.}
\algA is said to be admissible, iff with probability $1$, 
the two sequences 
it produces 
$\{x^{(0)}_t\}_{t\in [N]}$ and $\{x^{(1)}_t\}_{t \in [N]}$
are neighboring.
\end{tcolorbox}



The earlier works~\cite{continual,icalp10,tissec11} 
prove the following theorem about such
a DP prefix-sum mechanism:

\begin{theorem}[DP prefix-sum mechanism~\cite{continual,icalp10,tissec11}]
There is an $\epsilon$-adaptive-DP prefix-sum mechanism satisfying the above 
syntax, and moreover,
\begin{enumerate}
\item 
for any fixed $t \in [N]$ and $\sigma \in (0, 1)$, 
with probability $1-\sigma$, 
$\err_t < O(\frac{1}{\epsilon} \cdot \log^{1.5} N \cdot \log\frac{1}{\sigma})$.
\item 
making $n$ updates to the prefix-sum mechanism takes total time $O(n)$, that is,
the average time 
per update is $O(1)$.
\end{enumerate}
\label{thm:psum}
\end{theorem}

Note that although the earlier works~\cite{continual,icalp10,tissec11}
stated only the non-adaptive version 
of the above theorem where the sequence is not chosen adaptively,  
it is not hard to see that their proofs actually
work for adaptive sequences too.

\ignore{
We note that the proofs in these papers~\cite{continual,icalp10,tissec11}
imply that the DP guarantee holds even when  
the sequence of values is chosen adaptively, 
after having observed the answers of the mechanism in previous time steps.  
}

\section{A Quasilinear-Time Scheme}


\ignore{
\subsection{Warmup}

\elaine{we probably don't need the warmup in the final paper, just mention
it in narative in the technical roadmap.}

\begin{mdframed}
\begin{center}
{\bf Differentially Private Densest Subgraph}
\end{center}

\paragraph{Parameters and notations.}
Let $G := (V, E)$ be the input graph.
\elaine{FILL: $\epsilon_0$, $\epsilon_1$, $\epsilon'$, $\delta_1$}

\paragraph{Algorithm.}
\begin{enumerate}[leftmargin=5mm,itemsep=1pt]
\item 
Every vertex $v \in V$ computes 
its noisy degree $D(v) = \dgr(v) + \Geom(e^{\epsilon_0/2})$
where $\dgr(v)$ denotes $v$'s true degree.
\item 
Every vertex $v \in V$ initializes a DP-prefix-sum algorithm
with the parameters $(\epsilon_1, \delta_1)$.
Henceforth we use ${\PSum}(v)$
denote the DP-prefix-sum instance for the vertex $v$; moreover,
the notation ${\PSum}(v)$
also denotes the current outcome of the algorithm ${\PSum}(v)$.
\item 
Let $S := V$ and $d_{\rm max} := 0$. Repeat the following until $S$ is empty:
\item 
\begin{enumerate}[leftmargin=5mm,itemsep=1pt]
\item 
find the vertex $v \in S$ whose $D(v) - \PSum(v)$ 
is the smallest; 
\item 
if $d_{\rm max} < D(v) - \PSum(v)$, 
then update $d_{\rm max} := D(v) - \PSum(v)$ and let $S^* := S$; 
\item 
remove $v$ from $S$; 
\item 
for each residual $u \in S$: 
if $(u, v) \in E$, input $1$ to ${\PSum}(u)$; else 
input $0$ to ${\PSum}(u)$.
\end{enumerate}
\item 
Output $S^*$ and 
$d^* := \min\left(\frac{E(S^*) + \Geom(\exp({\epsilon'}))}{|S^*|}, |S^*|\right)$.
\end{enumerate}
\end{mdframed}

\subsection{Efficient Construction}
}

\subsection{Detailed Construction}
\label{sec:mainalg}

\elaine{note: is it important that cnt + noise is not used in computing the next 
vertex to remove?}

We first describe an algorithm
that runs in time quasilinear in $m + n$ where $m$ denotes
the number of edges and $n$ denotes the number of vertices. 
The intuition behind our algorithm has been explained in Section~\ref{sec:roadmap}.
Later in Section~\ref{sec:linear}, we describe how to improve
the algorithm's runtime to $O(m + n)$.

\SetKwInOut{Parameter}{Data}
\SetKwInOut{Remark}{Remark}

\begin{algorithm} [h]
\Remark{Below is the meta-algorithm. 
Immediately after the meta-algorithm description, we describe additional data structure
tricks to run it in quasilinear time.}
 \Parameter{Let $G := (V, E)$ be the input graph.
Let $\epsilon_0 = \epsilon_1 = \epsilon_2 = \epsilon' = \epsilon/4$. 
Let $\thresh := \frac{C}{\epsilon} \log n \log \frac1\sigma$ for a 
suitably large constant $C$.}
 \begin{algorithmic} [1]
 \STATE  Every vertex $v \in V$ computes 
its noisy degree $D(v) = \dgr(v) + \Geom(e^{\epsilon_0/2})$
where $\dgr(v)$ denotes $v$'s true degree.
\label{step:noisydeg}
\STATE 
Every vertex $v \in V$ initializes an $\epsilon_1$-DP prefix-sum algorithm.
\elaine{note: i made it epsilon dp here}
Henceforth we use ${\PSum}(v)$
denote the DP-prefix-sum instance for the vertex $v$; moreover,
the notation ${\PSum}(v)$
also denotes the current outcome of the algorithm ${\PSum}(v)$.
\STATE
\label{step:initcnt}
Every vertex $v \in V$ initializes 
a counter, denoted $\Cnt(v) := 0$.
Additionally, initialize a fresh noise $\noise(u) := \Geom(e^{\epsilon_2})$.
\label{step:initpsum}
\STATE Let $S := V$ and $d_{\rm max} := 0$. Repeat the following until $S$ is empty:
\label{step:loop}
\begin{enumerate}[label=(\alph*),leftmargin=5mm,itemsep=1pt]
\item 
Find the vertex $v \in S$ whose $D(v) - \PSum(v)$ 
is the smallest.
\item 
If $d_{\rm max} < D(v) - \PSum(v)$, 
then update $d_{\rm max} := D(v) - \PSum(v)$ and let $S^* := S$.
\item 
Remove $v$ from $S$. 
\item 
For each $u \in S$ such that $(u, v) \in E$:
let $\Cnt(u) := \Cnt(u) + 1$.
\item For each $u \in S$:
\begin{itemize}[leftmargin=5mm]
\item Let $\noisep := \Geom(e^{\epsilon_2})$ be a fresh noise.
\item If $\Cnt(u) + \noise(u) + \noisep > \thresh$: 
input $\Cnt(u)$ to $\PSum(u)$, 
reset $\Cnt(u) := 0$ and initialize a fresh noise $\noise(u) := \Geom(e^{\epsilon_2})$.
\end{itemize}
\end{enumerate}
\RETURN $S^*$ and 
$d^* := \min\left(\frac{E(S^*) + \Geom(\exp({\epsilon'}))}{|S^*|}, |S^*|\right)$.
 \end{algorithmic}
\caption{Differentially Private Densest Subgraph - Quasilinear-Time Variant.}
\label{alg:quasi}
\end{algorithm}

\paragraph{Running it in quasilinear time.}
In the above algorithm, if we run Line (4e) na\"ively as is, \elaine{hard coded ref}
it will require $\Omega(n^2)$ time.
However, with an additional data structructure trick, we can 
run the above algorithm in time $O(m + n \log n)$.
Recall that every time a vertex $u$'s outstanding counter $\Cnt(u)$ is reset to $0$, 
we also sample a fresh $\mcal{E}(u)$, and subsequently until $\Cnt(u)$ is reset
to $0$ again, in every time step, we will 
check if $\Cnt(u) + \mcal{E}(u) + \noisep > {\sf T}$ 
where $\noisep$ is freshly sampled in the respective time step --- and if so,
$\PSum(u)$ must be updated.
Equivalently, we can change the sampling 
as follows: any time 
$\Cnt(u)$ gets updated (i.e., either reset to $0$
or incremented), we sample  
a random variable ${\boldsymbol \tau}(u)$, which means the following:
if $\Cnt(u)$ does not get updated, when  
is the next time step in which the event $\Cnt(u) + \mcal{E}(u) + \noisep > {\sf T}$
happens. Note that if ${\boldsymbol \tau}(u) > n$, we can simply treat 
${\boldsymbol \tau}(u) = \infty$.
We may assume that ${\boldsymbol \tau}(u)$ can be sampled in constant time, 
since it follows a geometric distribution 
(see earlier works~\cite{dpfinite} on how to do this).
Therefore, we can maintain a table 
$L[1..n]$ where $L[t']$ stores a linked list of 
vertices that want their \PSum   
updated in time step $t'$. 
If a vertice $u$'s ${\boldsymbol \tau}(u)$ value changes 
from $t_1$ to $t_2$ before time step $t_1$,  
then we remove $u$ from $L[t_1]$ and add $u$ to $L[t_2]$ ---
this can be accomplished in constant time if 
in the table ${\boldsymbol \tau}(u)$ that stores the 
next \PSum update time for $u$, we also store a pointer to $u$'s position
in the table $L$.
Instead of executing Line 4(e) in a brute-force way, 
we can simply read $L[t]$ in each time step
$t$ to look for the vertices that want their \PSum updated in time step $t$.

\ignore{
Note that if $\Cnt(u)$ is being reset to $0$, we sample the new 
$\mcal{E}(u)$ first and then sample ${\boldsymbol \tau}(u)$.
In this way, we can maintain a binary heap data structure  
using ${\boldsymbol \tau}(u)$ as the comparison key for all residual vertices.
Instead of running Line 4(e) in a brute-force way, we instead
query the minimum ${\boldsymbol \tau}(u)$ values from the heap, to find
those whose \PSum must be updated in this time step.
Whenever ${\boldsymbol \tau}(u)$ is updated, we delete the old entry for $u$
from the heap, and insert the new value. 
}

Moreover, we will also use a Fibonnaci heap 
to maintain the residual noisy degree (i.e., $D(v) - \PSum(v)$) of every vertex.
Due to the way ${\sf T}$ is chosen,  for a proper constant $C$, one can show through a standard Chernoff bound that 
with the probability of $1-\sigma$, 
there cannot be more than $3m$ \PSum updates in which
the true increment input to the \PSum instance is 0. 
This means that there cannot be more than $O(m)$ \PSum updates in total 
except with $\sigma$  probability.
Summarizing the above, the above algorithm can be executed 
in time $O(m  + n \log n)$ with $1-\sigma$ probability.

\elaine{may need to fit it under page limit}

\subsection{Differential Privacy and Utility Guarantees} 
\label{sec:dpproof}
We now formally state the algorithm's differential privacy guarantees
as well as utility.

\begin{theorem}[Differential privacy of the output]
The algorithm in Section~\ref{sec:mainalg}
satisfies $(\epsilon_0 + \epsilon_1 + \epsilon_2 + \epsilon')$-DP.
\label{thm:dp}
\end{theorem}

Our algorithm's utility is stated in the following theorem:
\begin{theorem}[Utility of our algorithm]
Our algorithm in Section~\ref{sec:mainalg}
achieves $(2, O(\frac{1}{\epsilon} \cdot \log^{2.5}n \cdot \log\frac1\sigma))$-approximation with probability $1-\sigma$.
\label{thm:util}
\end{theorem}

\section{A Linear-Time Algorithm}
\label{sec:linear}

\begin{figure} [t!]
\centering
\begin{subfigure} [b] {0.49\textwidth}
  \includegraphics[width=\textwidth]{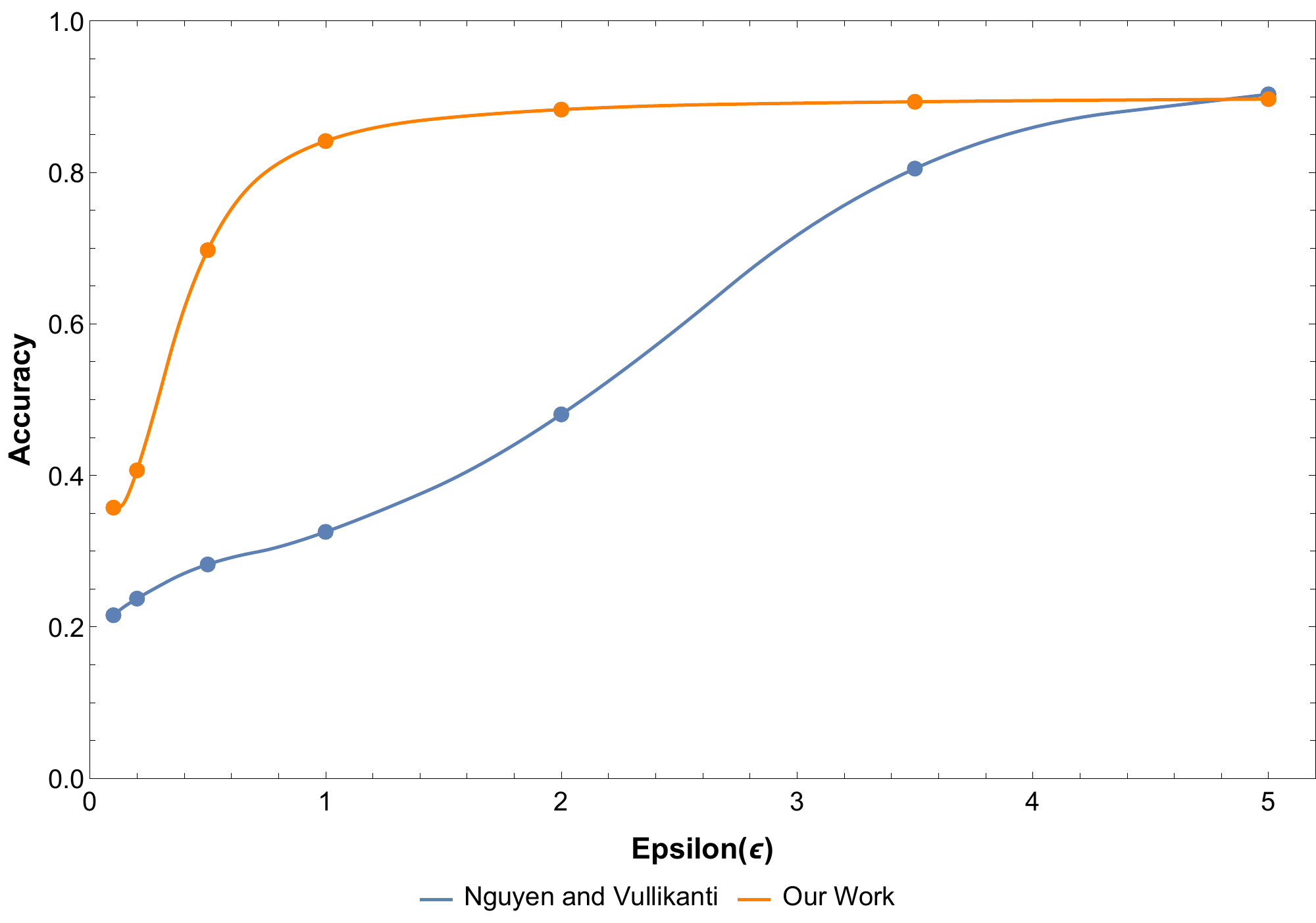}
  \vspace{-0.5cm}
  \caption{ca-Astro Network}
 \end{subfigure}
 \begin{subfigure} [b] {0.49\textwidth}
  \includegraphics[width=\textwidth]{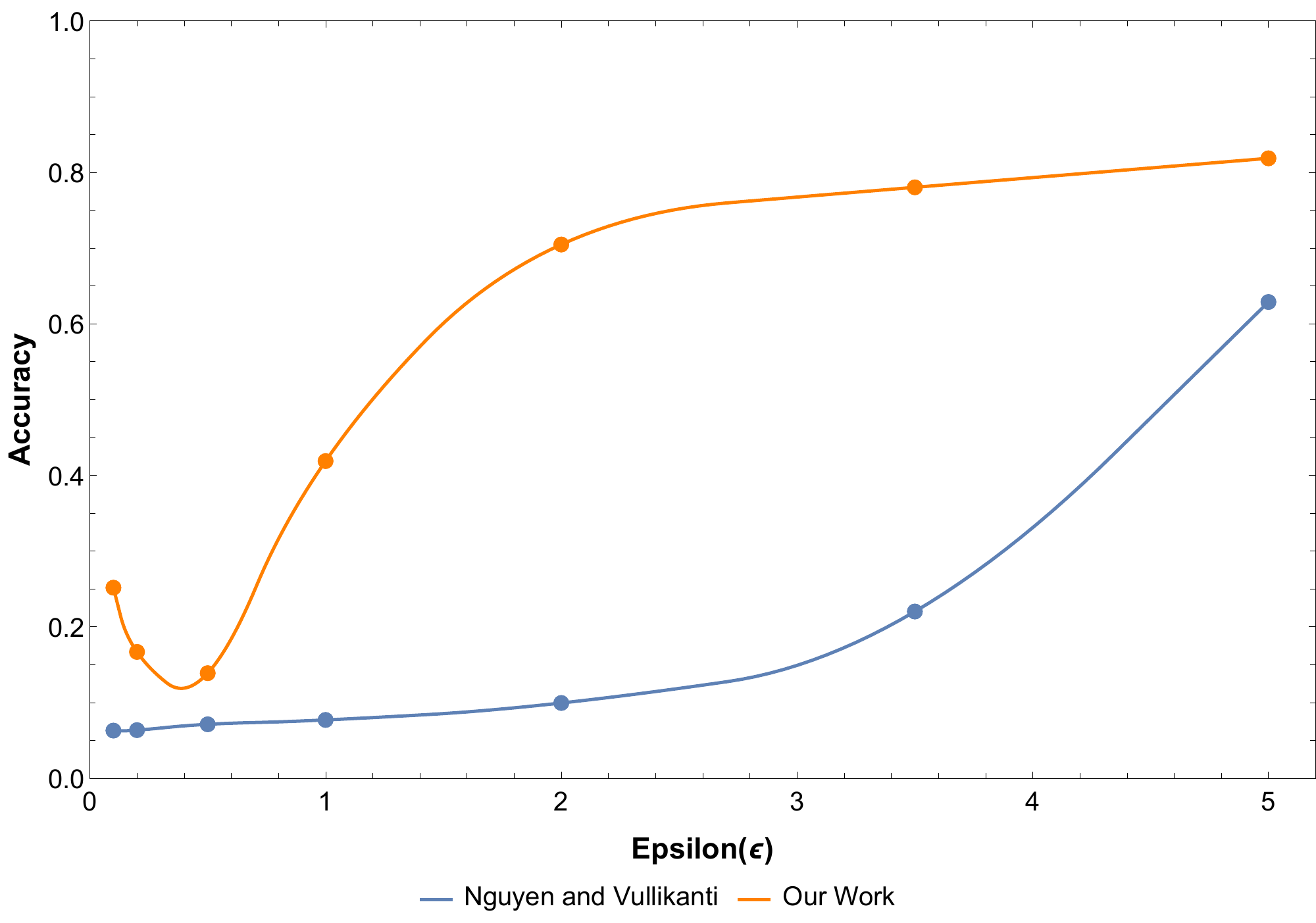}
  \vspace{-0.5cm}
  \caption{ca-GrQc Network}
 \end{subfigure}
  \begin{subfigure} [b] {0.49\textwidth}
  \includegraphics[width=\textwidth]{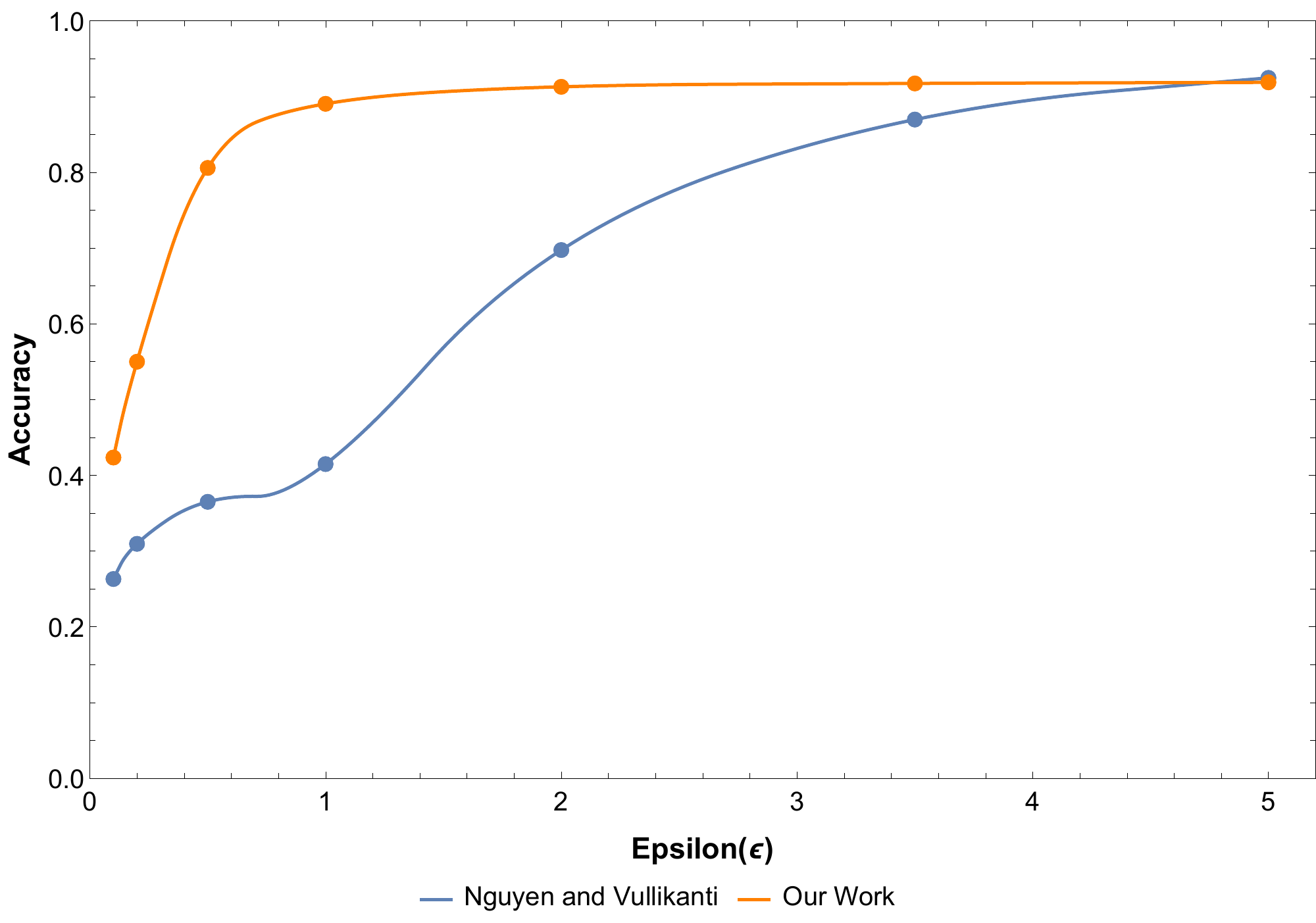}
  \vspace{-0.5cm}
  \caption{musae\underline{ }DE Network}
 \end{subfigure}
  \begin{subfigure} [b] {0.49\textwidth}
  \includegraphics[width=\textwidth]{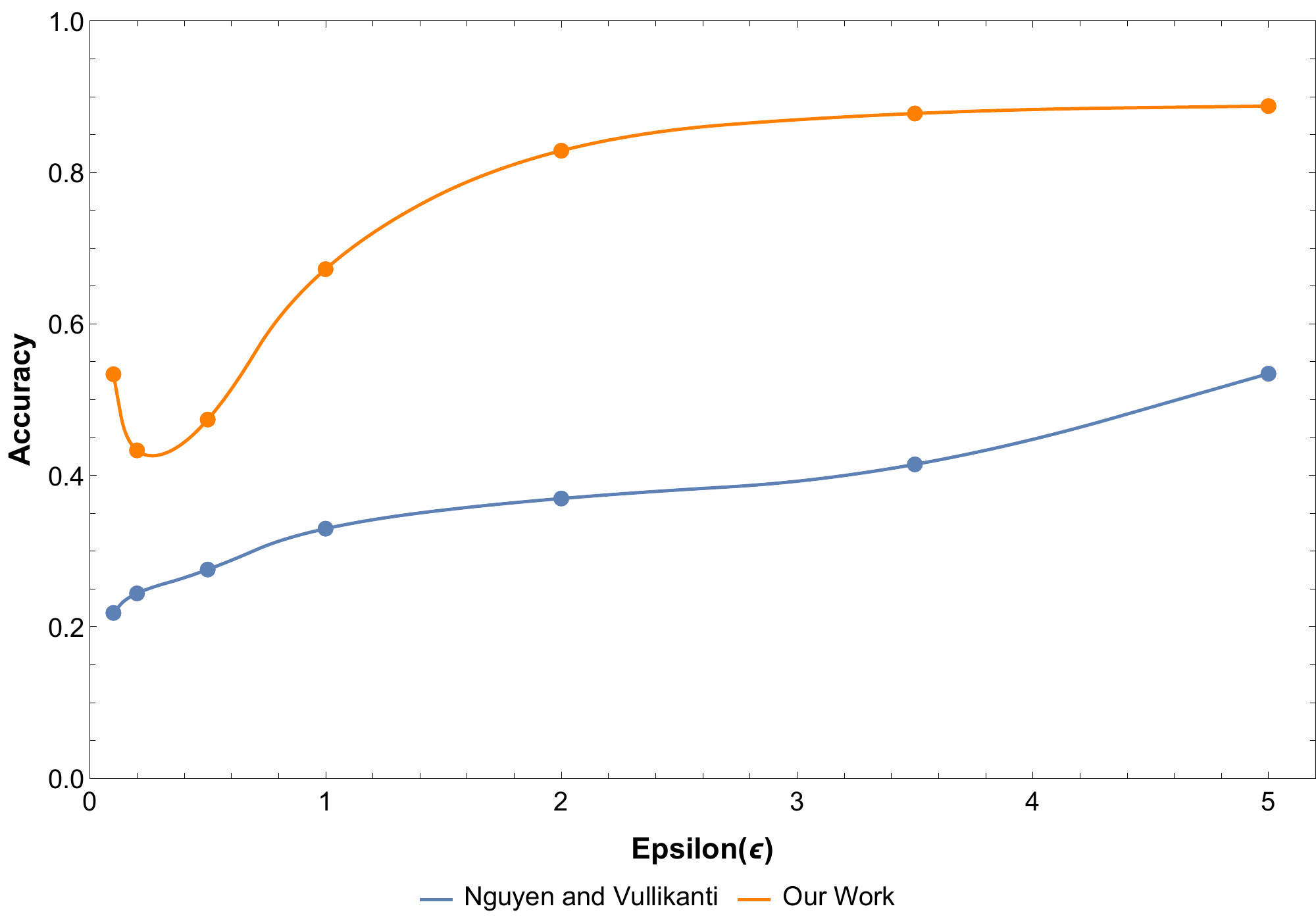}
  \vspace{-0.5cm}
  \caption{musae\underline{ }ENGB Network}
 \end{subfigure}
  \begin{subfigure} [b] {0.49\textwidth}
  \includegraphics[width=\textwidth]{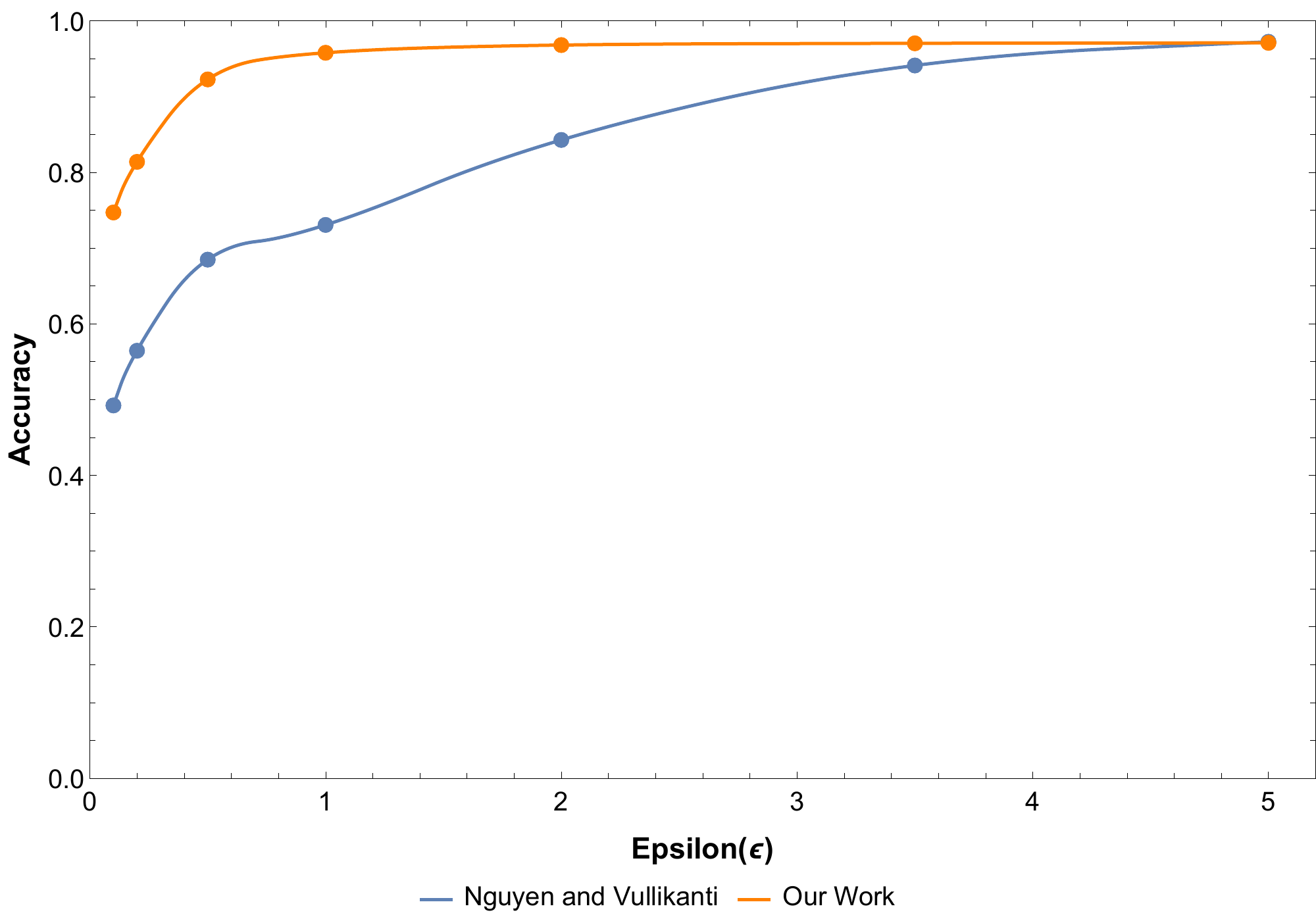}
  \vspace{-0.5cm}
  \caption{socfb-Amherst41 Network}
 \end{subfigure}
   \begin{subfigure} [b] {0.49\textwidth}
  \includegraphics[width=\textwidth]{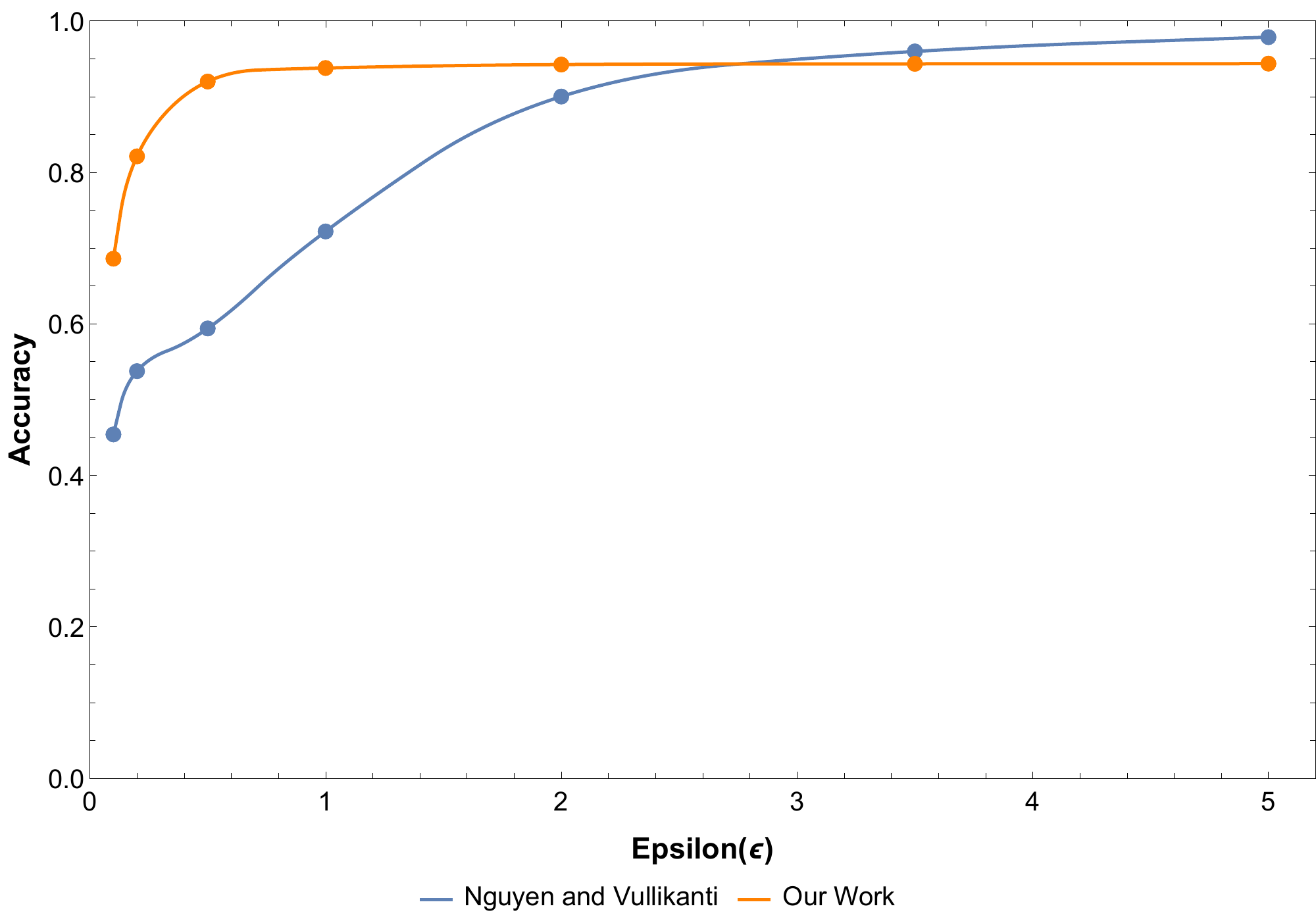}
  \vspace{-0.5cm}
  \caption{socfb-Auburn71 Network}
 \end{subfigure}
 \caption{ Accuracy of our linear-time algorithm as well as the accuracy of the Nguyen and Vullikanti \cite{nguyen2021differentially}. The orange line shows the accuracy of our proposed algorithm, and the blue line shows the accuracy of Nguyen and Vullikanti.
 }
 \label{fig:accuracy}
\end{figure}

As mentioned, the algorithm in Section~\ref{sec:mainalg}
takes $O(m + n\log n)$ time.
We suggest an improved version that runs in linear time, inspired
by a trick suggested by Charikar~\cite{charikardensestsubgraph}.


Observe that $\bigl| \min_{v \in S_{t+1}} |E(v, S_{t+1})|
 - \min_{v \in S_t} |E(v, S_t)| \bigr| \leq 1$.
Our algorithm is maintaining a noisy estimate $D(v) - \PSum(v)$ of 
$|E(v, S_t)|$ for all residual vertices $v$.
As we discussed, all estimates
have at most $O(\frac1\epsilon \log^{2.5}n \log\frac1\sigma)$ error
with $1-\sigma$ probability.
Henceforth let $C$ be a sufficiently large constant,
and let $\err := \frac C \epsilon \log^{2.5}n \log\frac1\sigma$. 
We can discretize the value of $D(v) - \PSum(v)$ into 
$B := \lceil n/\err + 1 \rceil$ buckets.
The $i$-th bucket will contain values 
from the range $[(i-1) \cdot \err - \err/2, (i-1)\cdot \err + \err/2]$.
We can now maintain a data structure such that each vertex
is placed in the right bucket depending on the current estimate $D(v) - \PSum(v)$. 
Inside each bucket, we maintain a linked list of vertices.
We also maintain an array such that each vertex stores 
the bucket it currently belongs to, as well as 
its pointer in the corresponding linked list.

With this idea, we can modify the algorithm in Section~\ref{sec:mainalg}
into a variant that runs in linear time.

\SetKwInOut{Init}{Initialization}
\begin{algorithm} [h]
 \Parameter{Let $G := (V, E)$ be the input graph.
Let $\epsilon_0 = \epsilon_1 = \epsilon_2 = \epsilon' = \epsilon/4$. 
Let $\thresh := \frac{C}{\epsilon} \log n \log \frac1\sigma$ for a 
suitably large constant $C$.}
\Init{Run the initialization steps of Algorithm \ref{alg:quasi} in Section~\ref{sec:mainalg}, 
and moreover, place all vertices in the right bucket.
Recall that each vertex stores a pointer to where it  
resides in the linked list of its bucket. 
Initially, let $\idx = 1$.
We may assume that the (imaginary) 0-th 
 bucket is always empty.}
 \begin{algorithmic} [1]
\STATE Let $S := V$ and $d_{\rm max} := 0$. Repeat the following until $S$ is empty:
\label{step:loop2}
\begin{enumerate}[label=(\alph*),leftmargin=5mm,itemsep=1pt]
\item 
Let $\idx'$ be the bucket from which a vertex is removed
in the previous time step. 
Check whether 
each bucket $\idx' - 1, \idx'$, $\cdots$ is non-empty, and 
let $\idx$ be the smallest non-empty bucket among these.
Let $v$ be an arbitrary vertex from this bucket $\idx$.
Remove $v$ from the bucket.
\item 
If $d_{\rm max} < D(v) - \PSum(v)$, 
then update $d_{\rm max} := D(v) - \PSum(v)$ and let $S^* := S$.
\item 
Remove $v$ from $S$. 
\item 
For each $u \in S$ such that $(u, v) \in E$:
let $\Cnt(u) := \Cnt(u) + 1$.
\item 
For each $u \in S$, 
if $\Cnt(u) + \noise(u) + \noisep > \thresh$ where 
$\noisep := \Geom(e^{\epsilon_2})$ denotes a fresh noise:\\
 \tcp{This line
is executed efficiently using the same data structure as mentioned in 
Section~\ref{sec:mainalg}.}
\begin{itemize}[leftmargin=5mm]
\item 
Input $\Cnt(u)$ to $\PSum(u)$;
\item 
Based on the updated value $D(u) - \PSum(u)$, relocate $u$ 
to a new bucket if necessary;
\item 
Reset $\Cnt(u) := 0$, and 
resample a fresh $\noise(u) := \Geom(e^{\epsilon_2})$.
\end{itemize}
\end{enumerate}
\RETURN $S^*$ and 
$d^* := \min\left(\frac{E(S^*) + \Geom(\exp({\epsilon'}))}{|S^*|}, |S^*|\right)$.
 \end{algorithmic}
\caption{Differentially Private Densest Subgraph -  Linear-Time Variant.}
\label{alg:linear}
\end{algorithm}

\begin{theorem}[Linear-time variant]
The above algorithm satisfies $\epsilon$-DP and 
moreover, it achieves 
 $(2, O(\frac{1}{\epsilon} \cdot \log^{2.5}n \cdot \log\frac1\sigma))$-approximation 
with probability $1-\sigma$.
The algorithm completes in time $O(n + m)$ with $1-\exp(-\Omega(m))$ probability
where $n = |V|$ and $m = |E|$.
\label{thm:linear}
\end{theorem}

\ignore{
For each time step $t$, do the following:
\begin{itemize}
\item 
Suppose that in the last time step $t-1$, a vertex is removed from bucket $i$.
Then, in the present time step $t$,  
check if buckets $i-1$, $i$, $i+1$ are empty, and pick
the smallest one that is not empty. 
Remove an arbitrary vertex denoted $v_t$ from this bucket, and $v_t$
will be the next vertex to remove from the residual graph $S$.

\item 
Update the outstanding counters $\Cnt(u)$ 
\end{itemize}
}

\section{Empirical Results}
\label{sec:exp}
In this section, we analyze our differentially private algorithm experimentally on real-world datasets and compare it to Charikar's algorithm and other DP algorithms for the densest subgraph problem. In our experiments, We use Charikar's algorithm as a non-differential private baseline. Table \ref{tab:stats} shows the 6 different networks that we used to evaluate the performance of our algorithm. For a graph $G$, let $S_c$ be the subgraph returned by Charikar's algorithm. We measure the accuracy of a DP algorithm as $\frac{\ds_G(S^*)}{\ds_G(S_c)}$ where $S^*$ is the subgraph returned by the DP algorithm. In other words, we measure the accuracy of the algorithms based on their relative performance in comparison to Charikar's algorithm.

\begin{table}
\centering
\begin{tabular}{ l c c c l }
\textbf{Network} & \textbf{Nodes} & \textbf{Edges} & \textbf{Baseline Solution} & \textbf{Description} \\
ca-Astro \cite{snapnets}& 18771 & 198050 & 29.554 & Collaboration net. arXiv Astro Phy. \\
ca-GrQc \cite{snapnets}& 5241 & 14484 & 22.3913 & Collaboration net. arXiv General Rel. \\
musae\underline{ }DE \cite{snapnets}& 9498 & 153138 & 39.0157 & Social net. of Twitch (DE) \\
musae\underline{ }ENGB \cite{snapnets}& 7126 & 35324 & 11.9295 & Social net. of Twitch (GB) \\
socfb-Amherst41 \cite{nr} & 2235 & 90954 & 54.5489 & Social net. of Facebook (Amherst) \\
socfb-Auburn71 \cite{nr} & 18448 & 973918 &  84.9551 & Social net. of Facebook (Auburn)
\end{tabular}
\caption{\label{tab:stats}
General statistics of networks}
\end{table}

Figure \ref{fig:accuracy} shows the accuracy of our linear-time algorithm as well as the accuracy of the $(\epsilon, \delta)$-DP algorithm proposed by Nguyen and Vullikanti \cite{nguyen2021differentially}. In our experiments, we have set the parameter $\sigma$  equal to $2^{-30} \approx 10^{-9}$ for our algorithm. We have also set $\delta = 10^{-9}$ for the $(\epsilon, \delta)$-DP algorithm proposed by Nguyen and Vullikanti. As it is shown in the figure, our proposed algorithm finds a much more accurate solution when $\epsilon$ is small. However, it achieves almost the same accuracy for large choices of $\epsilon$. Note that we have used $\textsc{SEQDENSEDP}$ algorithm to measure the accuracy of the work of Nguyen and Vullikanti which has better accuracy than their parallel algorithms \cite{nguyen2021differentially}.

We also study the running time linear time algorithm and compare it with the running time of $\textsc{SEQDENSEDP}$ algorithm proposed by Nguyen and Vullikanti. Table \ref{tab:runningtime} shows the average running time of the algorithms where the average is taken over 100 trials. Our experimental results show that in the RAM model, our linear time algorithm is much faster than the $\textsc{SEQDENSEDP}$ algorithm. In fact, our linear time algorithm is around 100 times faster than the $(\epsilon,\delta)$-DP algorithm of Nguyen and Vullikanti.

\begin{table}
\centering
\begin{tabular}{ l c c }
\textbf{Network} & \textbf{Our work} & \textbf{Nguyen and Vullikanti} \\
ca-Astro & 0.069 & 12.48 \\
ca-GrQc & 0.003 & 0.96 \\
musae\underline{ }DE & 0.034 & 3.36  \\
musae\underline{ }ENGB & 0.007 & 1.8  \\
socfb-Amherst41 & $<$0.001 & 0.21  \\
socfb-Auburn71 & 0.137 & 12.24
\end{tabular}
\caption{\label{tab:runningtime} Running time of DP algorithms. All running times are reported in seconds.}
\end{table}


\bibliographystyle{plainurl}
\bibliography{refs,bibdiffpriv,crypto}

\newpage
\appendix
\section{Failed Na\"ive Approaches}
\label{sec:naive}
A standard technique for attaining differential privacy
is to add noise to the answer calibrated
to the global sensitivity of the function being computed~\cite{calibrating}.
We argue that this approach fails to give meaningful
utility since the densest subgraph problem has high global sensitivity. 

\paragraph{Reporting the densest subgraph has high global sensitivity.}
The densest subgraph problem has high global sensitivity if we require
that the algorithm outputs the set of vertices that form a dense subgraph. 
For example, one can easily construct a family of graphs $G$ 
with $n$ vertices, and satisfying the following properties:
\begin{itemize}[leftmargin=5mm]
\item 
$G$ contains two disjoint sets of vertices $S$ 
and $S'$ of densities $d$ and $d'$ respectively.
There are no edges between $S$ and $S'$.
\item 
The densest subgraph of $S$ is $S$ itself;
and the densest subgraph of $S'$
is $S'$ itself. Furthermore, $d <  d' < d + 1/|S'|$.
\end{itemize}
This means that 
the densest subgraph of $G$ is $S'$; however, if we remove one edge
from $S'$, the densest subgraph of the resulting graph would become $S$.
Therefore, the ordinary approach of perturbing the output 
with noise roughly proportional to global sensitivity~\cite{calibrating}
does not apply.

\paragraph{Randomized response gives poor utility.}
A na\"ive approach for solving
the DP densest subgraph problem  is to rely on 
randomized response~\cite{randomizedresponse,dpbook-dr,dpbook-salil} 
where we use $G = (V, E)$ to denote the original graph given as input:
\begin{enumerate}[itemsep=1pt]
\item 
First, generate a rerandomized graph $\widetilde{G}$ as follows:
for each $(i, j) \in V^2$ and $i < j$, flip the existence of the edge $(i, j)$
with probability $p := \frac{1}{1 + e^{\epsilon}}$. 
\item 
Then, run an exact densest subgraph algorithm on 
the rerandomized graph $\widetilde{G}$.
Output the densest subgraph $S^* \subseteq V$ found. 
\item 
Output $d^* := (|E(S^*)| + \Geom(e^\epsilon))/|S^*|$ as an estimate
of $d(G)$.  
\end{enumerate}


It is not hard to show that the above algorithm 
satisfies $2\epsilon$-DP, 
However, the algorithm fails to give a good approximation.
To understand why, consider the following example ---
henceforth we will use $d^G(\cdot)$ and
$d^{\widetilde{G}}(\cdot)$
to denote the density of a subset of vertices
in the original graph $G$ and the rerandomized graph $\widetilde{G}$, respectively.
Suppose that $\epsilon = 1$ and thus $p = \Theta(1)$. 
Consider some graph $G = (V, E)$, in which the true densest
subgraph $S$ is a clique of size $\sqrt{n}$, and therefore
$d(G) = O\big(\sqrt{n}\big)$.
Moreover,
suppose that all vertices not in $S$ do not have any edges in $G$.
One can show that for any fixed subset $U \subseteq V$ 
of size at least $\frac{n}{\log^2 n}$, 
with $1-o(1)$ probability, 
$d^{\widetilde{G}}(U) \in [p \cdot |U|/2, 2p \cdot |U|]$. 
Thus with $1-o(1)$ probability, 
the na\"ive algorithm 
will report a large subgraph $U^* \approx V$ containing  
almost all the vertices.
However, when $U^* \approx V$, the true density $d^G (U^*) = O(1)$ which is 
a $O(\sqrt{n})$ factor 
smaller than the true answer $d(G)$. 
In other words, the reported set $U^*$ is not a good approximation
of the densest subgraph of $G$.

\ignore{
One can show that with 
some small, constant probability, in the rerandomized
graph $\widetilde{G}$, 
the vertices 
$V\backslash S$
will have more edges 
than the vertices 
}

\section{Deferred Proofs for our Quasilinear-Time Scheme}
\subsection{Deferred Proofs of Differential Privacy}
\label{sec:dpproof}

Below, we prove Theorem~\ref{thm:dp}.

Fix two arbitrary 
neighboring graphs $G$ and $\widetilde{G}$ that differ in only one edge.
In the proof below, we use the notation $\Pr[\cdot]$ 
to denote the probability 
when $G$ is used as the input, 
and we 
the notation $\widetilde{\Pr}[\cdot]$ to denote the probability 
when $\widetilde{G}$
is used as input.
In our algorithm, every time a vertex is removed from the residual set $S$,
we call it a time step denoted $t$.
Let $U_t$ be the set of vertices whose $\PSum$ instance is updated during time step $t$.
For a vertex $u \in U_t$, let $\Sigma_t(u)$
be the new output of $\PSum(u)$ after the update.
\elaine{
Observe that 
$\Cnt^G_t(u)$ is uniquely determined  
by XXX}
Henceforth, we use the notation $\prefix_t^u$ to denote 
the time steps up to $t$ (inclusive) in which $\PSum(u)$ is updated, 
the increment passed as input to $\PSum(u)$ during each of these updates,
and the new outcome of $\PSum(u)$ after each update.
We use $I_u$ to denote the time steps in which $u$ is updated.

\ignore{j
Observe that given 
$\{D(v)\}_{v \in V}$ and $(t_1, u_1, C_1, \Sigma_1), 
\ldots$, $(t_{j-1}, u_{j-1}, C_{j-1}, \Sigma_{j-1}), 
t_j, C_j, u_j$, 
the following are uniquely determined:
\begin{itemize}[leftmargin=5mm,itemsep=1pt]
\item 
the sequence of vertices removed from $S$ 
after the $(j-1)$-th \PSum update and till the $j$-th \PSum update; 
\item the next increment $\Cnt(u)$ passed as input to $\PSum(u)$ 
during the $j$-th \PSum update.
\end{itemize}
}

For an execution of the algorithm, we define the trace as $\{D(v)\}_{v \in V}, \{U_{t}, \{\Sigma_{t}(u)\}_{u \in U_t}\}_t$, where $t$ is different time steps of the algorithm. Note that given the trace, one can uniquely determine the sequence of the vertices removed in the densest subgraph algorithm. Throughout the proof, we fix an arbitrary trace $\tr = \big(\{D(v)\}_{v \in V}, \{U_{t}, \{\Sigma_{t}(u)\}_{u \in U_t}\}_t\big)$, and we show that $e^{-\epsilon} \cdot \widetilde{\Pr}[\tr] \le \Pr[\tr] \le e^{\epsilon} \cdot \widetilde{\Pr}[\tr]$, hence the algorithm is $\epsilon$-DP.

First, consider
when $G$ is the input.
Let $\prefix_0 := \{D(v)\}_{v \in V}$ and for $t > 0$, let
$\prefix_t := (\prefix_{t-1}, U_{t}, \{\Sigma_{t}(u)\}_{u \in U_{t}})$. 
\ignore{
Further, define $\Cnt^G_j(v) := \Cnt^G_j(\prefix_{j-1}, t_j, v)$ to be 
the $\Cnt(v)$ value at the time of the $j$-th \PSum update. 
Let $I^u := \{j : u_j = u\}$ be the set of indices 
that invoked an update to $\PSum(u)$, 
and let $I^{u}_j := \{k : u_k = u, k \leq j \}$.
We use $\prefix_j^u$ to denote the following event: 
\[
\text{Upon receiving the inputs $\{\Cnt^G_k(u_k)\}_{k \in I^{u}_j}$, 
$\PSum(u)$ outputs $\{\Sigma_k\}_{k \in I^{u}_j}$}
\]
}
We have the following: 
\begin{align*}
& 
\Pr[\{D(v)\}_{v \in V}, \{U_{t}, \{\Sigma_{t}(u)\}_{u \in U_t}\}_t]
\\ 
= & \Pr[\prefix_0] \cdot  
\prod_{t > 0} \bigg(
\Pr[U_{t} | \prefix_{t - 1}] \cdot  \\
&\ \  \quad \qquad \cdot \prod_{u \in U_{t}}
\Pr\left[\PSum(u) \text{ outputs } \Sigma_{t}(u) \text{ on new input } \Cnt^G_{t}(u)
\left| (\prefix_{t-1}, U_t)\right.
\right] \bigg)  \\
= & \Pr[\prefix_0] \cdot  
\prod_{t > 0} \bigg(
\Pr[U_t | \prefix_{t - 1}] \\ 
& \ \  \qquad \quad \cdot 
\prod_{u \in U_t}
\Pr\left[\PSum(u) \text{ outputs } \Sigma_t(u) \text{ on new input } \Cnt^G_t(u)
\left| \prefix_{t-1}^u \right.
\right] \bigg) \\
= & 
\Pr[\prefix_0] \cdot  
\prod_{t > 0} \Pr[U_t | \prefix_{t - 1}]
\cdot 
\prod_{u \in V} \Pr\left[\PSum(u) 
\text{ outputs } \{\Sigma_t(u)\}_{t \in I_u}
\text{ on updates } \{\Cnt^G_t(u)\}_{t \in I_u}
\right]
\end{align*}

Similarly, we have that 
\begin{align*}
& \widetilde{\Pr}[\{D(v)\}_{v \in V}, 
\left\{U_t,  \{\Sigma_t(u)\}_{u \in U_t}\right\}_{t}]
\\
=  & 
\widetilde{\Pr}[\prefix_0] \cdot
\prod_{t > 0} \widetilde{\Pr}[U_t | \prefix_{t - 1}]
\cdot
\prod_{u \in V} \widetilde{\Pr}\left[\PSum(u)
\text{ outputs } \{\Sigma_t(u)\}_{t \in I_u}
\text{ on updates } \{\Cnt^{\widetilde{G}}_t(u)\}_{t \in I_u}
\right]
\end{align*}

\begin{claim}
$e^{-\epsilon_0} \cdot {\Pr[\prefix_0]} \leq 
{\widetilde{\Pr}[\prefix_0]}
\leq e^{\epsilon_0} \cdot {\Pr[\prefix_0]}$
\label{clm:eps0}
\end{claim}
\begin{proof}
Observe that $G$ and $\widetilde{G}$ differ in 
at most one edge, and the (non)-existence of every 
edge affects the degree of at most two vertices. The claim therefore
follows from 
Fact~\ref{fct:geommech}.
\ignore{\[
e^{-\epsilon_0} \cdot {\Pr[\prefix_0]} \leq 
{\widetilde{\Pr}[\prefix_0]}
\leq e^{\epsilon_0} \cdot {\Pr[\prefix_0]}
\]}
\end{proof}

\begin{claim}
\begin{align*}
e^{-\epsilon_2} \cdot 
\prod_{t > 0} {\Pr}[U_t| \prefix_{t - 1}]
& \leq 
\prod_{t > 0} \widetilde{\Pr}[U_t | \prefix_{t - 1}]
\leq 
e^{\epsilon_2} \cdot \prod_{t > 0} {\Pr}[U_t| \prefix_{t - 1}]
\end{align*}
\label{clm:eps2}
\end{claim}
\begin{proof}
Recall that two graphs $G$ and $\widetilde{G}$ differ in only one edge. Let $e = \{i,j\}$ be this edge, then we can assume w.l.o.g. that $e \in G$ and $e \notin \widetilde{G}$. Now consider the trace $\{D(v)\}_{v \in V}, \{U_{t}, \{\Sigma_{t}(u)\}_{u \in U_t}\}_t$. As we discussed earlier given the trace, we can uniquely determine the sequence of the vertices removed in Step \ref{step:loop} of the algorithm. Let $t'$ be the first time step that the algorithm removes one of the vertices $i$ or $j$. By symmetry, we can assume that this vertex is $i$, i.e., $i \in U_{t'}$.

Assuming that the algorithm has the same set of noisy degrees $\{D(v)\}_{v \in V}$ at the beginning for both graphs $G$ and $\widetilde{G}$, the algorithm does not see any difference between $G$ and $\tilde{G}$ until it reaches time $t'$. This is because for every vertex $v$ that the algorithm removes before time $t'$, this vertex has the same set of neighbors in both $G$ and $\widetilde{G}$. Therefore,
\begin{align}
\label{eq:c1s}
\prod_{0 < t < t'} {\Pr}[U_t| \prefix_{t - 1}] = \prod_{0 < t < t'} \widetilde{\Pr}[U_t| \prefix_{t - 1}] \,.
\end{align}
Note that the equality above follows from the fact that given $\prefix_0$, the algorithm has the same set of noisy degrees $\{D(v)\}_{v \in V}$ for both $G$ and $\widetilde{G}$.

Consider the time step $t'$ where the algorithm removes $i$ from $S$. When we run the algorithm on $G$, vertex $j$ is a neighbor of $i$. Thus, the algorithm increases the $\Cnt(j)$ in this time step. However, this is not the case in $\widetilde{G}$, and the $\Cnt(j)$ remains the same when we run the algorithm on $\widetilde{G}$. We consider two cases for the rest of the proof.
\begin{itemize}
\item \textbf{Case 1:} The first case is when $j \notin U_{t}$ for any $t \ge t'$. This means that the algorithm never updates the $\PSum(j)$ after the step $t'$. Let $t''$ be the time step where the algorithm removes $j$ from $S$ in Step \ref{step:loop} of the algorithm. It is easy to see that given the trace of the algorithm, we can uniquely determine $t''$. We then have,
\begin{align}
\label{eq:c1g}
\prod_{t' \le t \le t''} {\Pr}[U_t| \prefix_{t - 1}] &= \prod_{t' \le t \le t''} \bigg(\prod_{u \in U_t} {\Pr}[u \in U_t| \prefix_{t - 1}] \cdot  \prod_{u \notin U_t} {\Pr}[u \notin U_t| \prefix_{t - 1}]\bigg) \nonumber \\
&= \prod_{u \in V} \bigg(\prod_{t \in [t',t''] \cap I_u } {\Pr}[u \in U_t| \prefix_{t - 1}] \cdot  \prod_{t \in [t',t''] - I_u} {\Pr}[u \notin U_t| \prefix_{t - 1}]\bigg) \,.
\end{align}
Similarly, for the graph $\widetilde{G}$, we have
\begin{align}
\label{eq:c1gt}
\prod_{t' \le t \le t''} \widetilde{\Pr}[U_t| \prefix_{t - 1}] &= \prod_{u \in V} \bigg(\prod_{t \in [t',t''] \cap I_u } \widetilde{\Pr}[u \in U_t| \prefix_{t - 1}] \cdot  \prod_{t \in [t',t''] - I_u} \widetilde{\Pr}[u \notin U_t| \prefix_{t - 1}]\bigg) \,.
\end{align}
In order to complete the proof of this case, we first claim that for every vertex $u \neq j$, we have
\begin{align}
\label{eq:c1m}
&\prod_{t \in [t',t''] \cap I_u } {\Pr}[u \in U_t| \prefix_{t - 1}] \cdot  \prod_{t \in [t',t''] - I_u} {\Pr}[u \notin U_t| \prefix_{t - 1}] \nonumber \\
&=\prod_{t \in [t',t''] \cap I_u } \widetilde{\Pr}[u \in U_t| \prefix_{t - 1}] \cdot  \prod_{t \in [t',t''] - I_u} \widetilde{\Pr}[u \notin U_t| \prefix_{t - 1}] \,.
\end{align}
The claim clearly holds for every vertex $u \neq i,j$, since $u$ has exactly the same set of neighbors in both $G$ and $\widetilde{G}$. Therefore, given $\prefix_{t-1}$ we have $\Cnt^G_{t}(u) = \Cnt^{\widetilde{G}}_t(u)$, thus the claim holds. Also, considering the vertex $i$, the algorithm removes $i$ from $S$ at the time step $t'$, and it never updates the prefix sum for $i$ after that. Therefore ${\Pr}[i \in U_t| \prefix_{t - 1}] = \widetilde{\Pr}[i \in U_t| \prefix_{t - 1}]= 0$ for any $t \ge t'$. Similarly, we have ${\Pr}[i \notin U_t| \prefix_{t - 1}] = \widetilde{\Pr}[i \notin U_t| \prefix_{t - 1}]= 1$ for any $t \ge t'$ which implies our claim for the vertex $i$. We now give the following bound for vertex $j$.

\begin{claim}
\label{clm:c1}
For  Case 1, let $p = \prod_{t \in [t',t''] \cap I_j } {\Pr}[j \in U_t| \prefix_{t - 1}] \cdot  \prod_{t \in [t',t''] - I_j} {\Pr}[j \notin U_t| \prefix_{t - 1}]$ and $\widetilde{p} = \prod_{t \in [t',t''] \cap I_j } \widetilde{\Pr}[j \in U_t| \prefix_{t - 1}] \cdot  \prod_{t \in [t',t''] - I_j} \widetilde{\Pr}[j \notin U_t| \prefix_{t - 1}]$,  we then have $e^{-\epsilon_2} \widetilde{p} \leq p \leq e^{\epsilon_2} \widetilde{p}$.
\end{claim}
\begin{proof}
Recall that we are assuming that $j \notin U_t$ for any $t\ge t'$. Therefore, $[t',t''] \cap I_j = \emptyset$, and we have
\begin{align*}
\prod_{t \in [t',t''] \cap I_j } {\Pr}[j \in U_t| \prefix_{t - 1}] \cdot  \prod_{t \in [t',t''] - I_j} {\Pr}[j \notin U_t| \prefix_{t - 1}] 
=  {\Pr}[j \notin \cup_{t' \le t \le t''} U_t | \prefix_{t' - 1}] \,.
\end{align*}
Similarly, we have
\begin{align*}
\prod_{t \in [t',t''] \cap I_j } \widetilde{\Pr}[j \in U_t| \prefix_{t - 1}] \cdot  \prod_{t \in [t',t''] - I_j} \widetilde{\Pr}[j \notin U_t| \prefix_{t - 1}] =  
\widetilde{\Pr}[j \notin \cup_{t' \le t \le t''} U_t | \prefix_{t' - 1}] \,.
\end{align*}

 Now considering the vertex $j$, 
the expression 
${\Pr}[j \notin \cup_{t' \le t \le t''} U_t | \prefix_{t' - 1}]$
(or $\widetilde{\Pr}[j \notin \cup_{t' \le t \le t''} U_t | \prefix_{t' - 1}]$)
can be equivalently thought as the following.
For every time step $t' \le t \le t''$, compute the 
noisy counter $\Cnt(j) + \noisep$ and compare it with 
the noisy threshold $\thresh - \noise(j)$,
where $\noisep$ is a fresh random noise, and
 $\noise(j)$ was chosen the last time $\Cnt(j)$ was reset to 0.
We want to know what is the probability that 
for all of $t' \le t \le t''$, $\Cnt(j) + \noisep$
never exceeds the threshold $\thresh - \noise(j)$.
This random process is identical to 
the sparse vector algorithm (see page 57, Algorithm 1 of \cite{dpbook-dr}),
applied to the following database and sequence of queries,
with the privacy budget $\epsilon_2$.
Specifically, the database here is a sequence of boolean values 
that represent whether $j$ is connected to the vertex being removed
in steps $t \in [t', t'']$.
The sequence of queries is whether the prefix sum of the database
in each time step exceeds ${\sf T}$. 
Note that for $G$ and $\widetilde{G}$ the two databases defined
above differ only in one position, i.e., the bit in time step $t'$ 
when $i$ is removed.
Moreover, all the prefix sums have sensitity $1$. 
The proof of the sparse vector technique immediately gives us
the following (see Theorem 3.23 of \cite{dpbook-dr}):
\begin{align*}
e^{-\epsilon_2} \cdot \widetilde{\Pr}[j \notin \cup_{t' \le t \le t''} U_t | \prefix_{t' - 1}] \leq 
{\Pr}&[j \notin \cup_{t' \le t \le t''} U_t | \prefix_{t' - 1}]\\
&\leq e^{\epsilon_2}
\cdot \widetilde{\Pr}[j \notin \cup_{t' \le t \le t''} U_t | \prefix_{t' - 1}] \,.
\end{align*}
\end{proof}

Considering any time step $t > t''$, the algorithm has removed both vertices $i$ and $j$, and it does not see any difference between $G$ and $\widetilde{G}$. Thus,
\begin{align*}
\prod_{t'' < t } {\Pr}[U_t| \prefix_{t - 1}] = \prod_{t'' < t } \widetilde{\Pr}[U_t| \prefix_{t - 1}] \,.
\end{align*}
The equality above along with equalities (\ref{eq:c1s}), (\ref{eq:c1m}) and also Claim \ref{clm:c1} implies Claim \ref{clm:eps2} for this case.

\item \textbf{Case 2:} The second case is when $j \in U_t$ for some $t\ge t'$. Let $t'' \ge t'$ be the smallest index such that $j \in U_{t''}$. 
It is easy to verify that equations (\ref{eq:c1g}) and (\ref{eq:c1gt}) still hold in Case 2. Also, Equation (\ref{eq:c1m}) holds for every vertex $u \neq j$. We now give the following bound for vertex $j$.

\begin{claim}
\label{clm:c2}
For  Case 2, let $p = \prod_{t \in [t',t''] \cap I_j } {\Pr}[j \in U_t| \prefix_{t - 1}] \cdot  \prod_{t \in [t',t''] - I_j} {\Pr}[j \notin U_t| \prefix_{t - 1}]$ and $\widetilde{p} = \prod_{t \in [t',t''] \cap I_j } \widetilde{\Pr}[j \in U_t| \prefix_{t - 1}] \cdot  \prod_{t \in [t',t''] - I_j} \widetilde{\Pr}[j \notin U_t| \prefix_{t - 1}]$,  we then have $e^{-\epsilon_2} \widetilde{p} \leq p \leq e^{\epsilon_2} \widetilde{p}$.
\end{claim}
\begin{proof}
Recall that we are assuming that $j \in U_{t''}$ and $j \notin U_t$ for any $t' \le t < t''$. Therefore, $[t',t''] \cap I_j = \{t''\}$, and we have
\begin{align*}
&\prod_{t \in [t',t''] \cap I_j } {\Pr}[j \in U_t| \prefix_{t - 1}] \cdot  \prod_{t \in [t',t''] - I_j} {\Pr}[j \notin U_t| \prefix_{t - 1}]\\
& =  {\Pr}[j \in U_{t''}| \prefix_{t'' - 1}] 
\cdot 
{\Pr}[j \notin \cup_{t' \le t < t''} U_t| \prefix_{t' - 1}] 
\quad \ \  (\text{henceforth denoted } {p})
\,.
\end{align*}
Similarly, we have
\begin{align*}
&\prod_{t \in [t',t''] \cap I_j } \widetilde{\Pr}[j \in U_t| \prefix_{t - 1}] \cdot  \prod_{t \in [t',t''] - I_j} \widetilde{\Pr}[j \notin U_t| \prefix_{t - 1}] \\
& =  \widetilde{\Pr}[j \in U_{t''}| \prefix_{t'' - 1}]
\cdot
\widetilde{\Pr}[j \notin \cup_{t' \le t < t''} U_t| \prefix_{t' - 1}] 
\quad \ \  (\text{henceforth denoted } \widetilde{p})
\,.
\end{align*}

The expressions 
$p$ or $\widetilde{p}$
are equivalent 
to the following.
For every time step $t' \le t \le t''$, the algorithm computes the noisy counter $\Cnt(j) + \noisep$ and compares it to the noisy threshold $\thresh - \noise(j)$. 
We want to know what is the probability that 
for all of $t' \le t < t''$, 
$\Cnt(j) + \noisep$ does not exceed threshold $\thresh - \noise(j)$,
but finally in time $t''$, 
$\Cnt(j) + \noisep$ indeed exceeds threshold $\thresh - \noise(j)$.
Similar to what we discussed in Claim \ref{clm:c1}, 
we can directly apply the analysis of the sparse
vector technique here, 
and obtain that 
$ e^{-\epsilon_2} \cdot \widetilde{p}
\leq  p \leq e^{\epsilon_2} \cdot \widetilde{p}$.

\end{proof}

Considering any time step $t > t''$, the algorithm has removed vertices $i$ and $j$, and the induced subgraph between the vertices in $S$ is exactly the same for both graphs $G$ and $\widetilde{G}$. The algorithm also resets the counters $\Cnt^G(j)$ and $\Cnt^{\widetilde{G}}(j)$ at the time step $t''$. Furthermore, the $\PSum$ values are exactly the same for $G$ and $\widetilde{G}$ conditioning on $\prefix_{t''}$. Thus, the algorithm does not see any difference between $G$ and $\widetilde{G}$ after the time step $t''$ and we have
\begin{align*}
\prod_{t'' < t } {\Pr}[U_t| \prefix_{t - 1}] = \prod_{t'' < t } \widetilde{\Pr}[U_t| \prefix_{t - 1}] \,.
\end{align*}
The equality above along with equalities (\ref{eq:c1s}), (\ref{eq:c1m}) and also Claim \ref{clm:c2} implies Claim \ref{clm:eps2} for this case.
\end{itemize}
This completes the proof for both Case 1 and Case 2 and proves Claim \ref{clm:eps2}.
\end{proof}
\begin{claim}
Let 
${p} := 
\prod_{u \in V} {\Pr}\left[\PSum(u)
\text{ outputs } \{\Sigma_t(u)\}_{t \in I^u}
\text{ on updates } \{\Cnt^{{G}}_t(u)\}_{t \in I^u}
\right]$
and let 
$\widetilde{p} := 
\prod_{u \in V} \widetilde{\Pr}\left[\PSum(u)
\text{ outputs } \{\Sigma_t(u)\}_{t \in I^u}
\text{ on updates } \{\Cnt^{\widetilde{G}}_t(u)\}_{t \in I^u}
\right]$.
It must be that 
$e^{-\epsilon_1} p \leq \widetilde{p} \leq e^{\epsilon_1} p$.
\label{clm:eps1}
\end{claim}
\begin{proof}
Since $G$ and $\widetilde{G}$ are neighboring, there is at most one 
$(u, t)$ pair where $t \in I_u$ 
such that $\Cnt^{\widetilde{G}}_t(u)$ and $\Cnt^{{G}}_t(u)$ differ by 1.
For all other $(u, t)$ pair where $t \in I_u$, 
$\Cnt^{\widetilde{G}}_t(u)$ and $\Cnt^{{G}}_t(u)$ must be the same. 
The claim therefore follows
from the $\epsilon_1$-DP of the prefix sum mechanism.

\paragraph{Proof of Theorem~\ref{thm:dp}.}
Let $S^*$ be the subgraph output by the algorithm. Observe that $S^*$ is uniquely determined by trace $\{D(v)\}_{v \in V}, 
\left\{U_t, \allowbreak \{\Sigma_t(u)\}_{u \in U_t}\right\}_{t}$.
Therefore, given Claims~\ref{clm:eps0}, \ref{clm:eps2}, and \ref{clm:eps1}, 
we conclude that for any $S^*$, 
$e^{-(\epsilon_0 + \epsilon_1 + \epsilon_2)} \cdot \Pr[S^*] 
\leq \widetilde{\Pr}[S^*]
\leq e^{\epsilon_0 + \epsilon_1 + \epsilon_2} \cdot \Pr[S^*]$.
Besides $S^*$, the algorithm also needs to output $d^*$. 
Since $G$ and $\widetilde{G}$ differ in at most one edge, 
and due to the distribution of the noise in computing $d^*$, 
it follows that for any $d^*$ and $S^*$, 
\[
e^{-\epsilon'} \Pr[d^* | S^*] \leq \widetilde{\Pr}[ d^* | S^*] \leq 
e^{\epsilon'} \Pr[d^* | S^*] 
\]
Therefore, 
for any $S^*$ and $d^*$, we have that
$e^{-(\epsilon_0 + \epsilon_1 + \epsilon_2 + \epsilon')} \cdot \Pr[S^*, d^*] 
\leq \widetilde{\Pr}[S^*, d^*]
\leq e^{\epsilon_0 + \epsilon_1 + \epsilon_2 + \epsilon'} \cdot \Pr[S^*, d^*]$.
\end{proof}

\subsection{Deferred Proofs of Utility}
\label{sec:utilproof}
In this section, we prove Theorem~\ref{thm:util}.

To prove Theorem~\ref{thm:util}, we will need to use the following lemma
proven by Charikar~\cite{charikardensestsubgraph}.
\begin{lemma}[Upper bound on $d(G)$~\cite{charikardensestsubgraph}]
Let $G := (V, E)$ be an undirected graph, and suppose that we arbitrarily 
assign an orientation to each edge. 
Let $d_{\rm max}$ be the maximum number of edges oriented towards any vertex.
Then, it must be that $d(G) \leq d_{\rm max}$.
\label{lem:dmax}
\end{lemma}

\begin{proof}{ of Theorem~\ref{thm:util}}

\begin{claim}
With $1- 0.1 \sigma$ probability, 
the following holds throughout the algorithm:
at the beginning of every time step, 
let $S$ be the residual set, and let $v \in S$; 
then, 
\[
\bigl| D(v) - \PSum(v) - |E(v, S)| \bigr| 
\leq O\left(\frac{1}{\epsilon} \cdot \log^{2.5} n \cdot \log\frac{1}{\sigma}\right) + \thresh.
\]
\label{clm:bounderr}
\end{claim}
\begin{proof}
Recall that $\epsilon_0 = \epsilon_1 = \epsilon_2 = \epsilon' = \epsilon/4$.
For a fixed $v \in V$, 
by the property of the $\Geom(e^{\epsilon_0/2})$ noise distribution in $D(v)$,  
with probability $1 - \frac{0.1\sigma}{2n}$, 
$|\deg(v) - D(v)| < \frac{C}{\epsilon} \cdot \log n 
\cdot \log \frac{1}{\sigma}$ for some appropriate constant $C$.
Taking the union bound over all $v \in V$, with probability 
$1- 0.1\sigma/2$, 
it holds that for all $v \in V$,
$|\deg(v) - D(v)| < \frac{C}{\epsilon} \cdot \log n
\cdot \log \frac{1}{\sigma}$.

By Theorem~\ref{thm:psum}, 
for any $v \in V$ and any fixed time step, 
with probability $1- \frac{0.1 \sigma}{2n^2}$, 
the error of $\PSum(v)$ for the fixed time step 
is upper bounded by 
$O(\frac{1}{\epsilon} \cdot \log^{2.5} n \cdot \log\frac{1}{\sigma})$.
Taking a union bound over all time steps, it must be that 
for any fixed $v$, with probability
$1-\frac{0.1 \sigma}{2n}$, 
the error of $\PSum(v)$ 
is upper bounded 
$O(\frac{1}{\epsilon} \cdot \log^{2.5} n \cdot \log\frac{1}{\sigma})$
in all time steps.
Taking a union bound over all vertices, it must be that 
with probability $1-\frac{0.1 \sigma}{2}$, the above holds for all vertices.

Recall that we do not update $\PSum(v)$ in every time step, only when 
$\Cnt(v) + \noise(v) + \noisep$ has exceeded the threshold \thresh.
Further, $E(v, V \backslash S)$ is equal to the true sum of all increments input 
to $\PSum(v)$ so far, plus $\Cnt(v)$.
For any fixed $\noise(v)$, with $1-0.1\sigma/n^3$ probability, 
$|\noise(v)| \leq O(\frac1\epsilon \log n \log \frac1\sigma)$.
The same holds for each fixed $\noisep$.
Therefore, with $1-0.1\sigma$ probability, it must be
that any noise $\noise(v)$ or $\noisep$ generated throughout  
the algorithm has magnitude at most $O(\frac1\epsilon \log n \log \frac1\sigma)$.
This means with $1-0.1\sigma$ probability, 
throughout the algorithm and for any $v$,  
the outstanding counter value $\Cnt(v)$ 
that has not been accumulated by $\PSum(v)$ 
cannot exceed $\thresh + O(\frac1\epsilon \log n \log \frac1\sigma)$.

Therefore, with probability $1-0.1\sigma$,
it must be that 
at the beginning of every time step, and for every $v \in S$
where $S$ is the current residual set --- henceforth,
we use ${\sf TruePSum(v)}$ to mean the true prefix sum 
of all inputs that have been sent to $\PSum(v)$:
\begin{align*}
& \bigl| D(v) - \PSum(v) - |E(v, S)| \bigr|\\
 = & 
\bigl| D(v) - \PSum(v) - (\deg(v) - |E(v, V \backslash S|) \bigr| \\
 \leq &
\bigl| D(v) - \deg(v)\bigr| + 
 \bigl||E(v, V \backslash S|) - \PSum(v)\bigr| \\
= &
\bigl| D(v) - \deg(v)\bigr| + 
 \bigl|   \Cnt(v) + {\sf TruePSum}(v) - \PSum(v)\bigr| \\
\leq &
\bigl| D(v) - \deg(v)\bigr| + 
 \bigl|   \Cnt(v) \bigr| + \bigl| {\sf TruePSum}(v) - \PSum(v)\bigr| \\
\leq  &
O\left(\frac{1}{\epsilon} \cdot \log n \cdot \log\frac1\sigma \right)
+ 
\thresh + 
O\left(\frac{1}{\epsilon} \cdot \log n \cdot \log\frac1\sigma \right)
+ 
O\left(\frac{1}{\epsilon} \cdot \log^{2.5} n \cdot \log\frac{1}{\sigma}\right)
\\
\leq & 
O\left(\frac{1}{\epsilon} \cdot \log^{2.5} n \cdot \log\frac{1}{\sigma}\right)
 + \thresh
\end{align*}

\end{proof}

\ignore{
\begin{claim}
With $1-\sigma$ probability, 
it holds that for all $v$, 
$|\psi(v) - \overline{\psi}(v)| \leq  
O(\frac{1}{\epsilon} \cdot \log^{2.5} n \cdot \log\frac{1}{\sigma})
$. 
\end{claim}
\begin{proof}
For any $v \in V$, 
$|\psi(v) - \overline{\psi}(v)| \leq  |\deg(v) - D(v)| + 
| E(v, V \backslash S) - \PSum(v)|$
where $S$ and $\PSum(v)$ denote the values of these variables 
at the time $v$ is being removed from $S$.

For a fixed $v \in V$, 
by the property of the $\Geom(e^{\epsilon_0/2})$ noise distribution in $D(v)$,  
with probability $1 - \frac{\sigma}{2n}$, 
$|\deg(v) - D(v)| < \frac{C}{\epsilon} \cdot \log n 
\cdot \log \frac{1}{\sigma}$ for some appropriate constant $C$.
Taking the union bound over all $v \in V$, with probability 
$1-\sigma/2$, 
it holds that for all $v \in V$,
$|\deg(v) - D(v)| < \frac{C}{\epsilon} \cdot \log n
\cdot \log \frac{1}{\sigma}$.

By Theorem~\ref{thm:psum}, 
for any $v \in V$ and any fixed time step, 
with probability $1- \frac{\sigma}{2n^2}$, 
the error of $\PSum(v)$ for the fixed time step 
is upper bounded by 
$O(\frac{1}{\epsilon} \cdot \log^{2.5} n \cdot \log\frac{1}{\sigma})$.
Taking a union bound over all time steps, it must be that with probability
$1-\frac{\sigma}{2n}$, 
the error of $\PSum(v)$ 
is upper bounded 
$O(\frac{1}{\epsilon} \cdot \log^{2.5} n \cdot \log\frac{1}{\sigma})$
in all time steps.
Taking union bound over the choice of $v$, it must be that
with probability $1-\sigma/2$, 
all $\PSum$ instances have error bounded by  
$O(\frac{1}{\epsilon} \cdot \log^{2.5} n \cdot \log\frac{1}{\sigma})$
in all time steps.

Therefore, with probability $1-\sigma$, 
$|\psi(v) - \overline{\psi}(v)| \leq  |\deg(v) - D(v)| +
| E(v, V \backslash S) - \PSum(v)| \leq 
O(\frac{1}{\epsilon} \cdot \log n \cdot \log\frac{1}{\sigma})
 + O(\frac{1}{\epsilon} \cdot \log^{2.5} n \cdot \log\frac{1}{\sigma})
= O(\frac{1}{\epsilon} \cdot \log^{2.5} n \cdot \log\frac{1}{\sigma})
$.
\elaine{the union bound loss can be bounded tighter}
\end{proof}
}


\begin{claim}
Consider some execution of our algorithm, 
and let 
and let $\err := \max_{t, v \in S_t} \bigl| D(v) - \PSum_t(v) - |E(v, S_t)| \bigr|$
where $S_t$ denotes the residual set at the beginning of time step $t$,
and $\PSum_t(v)$ denotes the output of $\PSum(v)$ at the beginning of time step $t$.
Then, $d(S^*) \geq (d(G) - 4 \cdot \err)/2$.
\label{clm:approxanalysis}
\end{claim}
\begin{proof}
Fix an arbitrary time step $t$, 
it holds that 
$\min_{v \in S_t} |E(v, S_t)| \leq 2 | E(S_t) |/|S_t|$. 
Let $v_t \in S_t$ be the actual vertex that is removed in time step $t$, 
i.e., 
$v_t := {\rm arg}\min_v(D(v)-\PSum_t(v))$.
It therefore holds that
$|E(v_t, S_t)| - \min_{v \in S_t} |E(v, S_t)| \leq 2 \cdot \err$.

Now, suppose that as a vertex $v_t$ gets removed from the residual graph
$S_t$, all edges $E(v_t, S_t)$ are oriented towards $v_t$.
Observe that the $d_{\rm max} := \max_t( D(v_t) - \PSum_t(v_t) )$
value at the end of the algorithm is a good estimate of 
$\max_t(E(v_t, S_t))$.
Specifically, 
$d_{\rm max} \geq \max_t(E(v_t, S_t)) - 2 \cdot \err$.
Henceforth let $t^* := {\rm arg}\max_t( D(v_t) - \PSum_t(v_t) )$,
and therefore, our algorithm's output $S^* := S_{t^*}$.

By Lemma~\ref{lem:dmax}, it holds that 
$d(G) \leq \max_t(E(v_t, S_t)) \leq d_{\rm max} + 2 \cdot \err 
= D(v_{t^*}) - \PSum_{t^*}(v_{t^*})
+ 2 \cdot \err \leq  
\min_{v \in S_{t^*}} |E(v, S_{t^*})| + 2 \cdot \err + 2 \cdot \err
= 2 \cdot \err \leq  
2|E(S_{t^*})|/|S_{t^*}|  + 4 \cdot \err
= 2|E(S^*)|/|S^*|  + 4 \cdot \err
$.
In other words, $d(S^*) \geq (d(G) - 4 \cdot \err)/2$.
\end{proof}

Finally, 
observe that by the definition of $d^*$ and due to Fact~\ref{fct:geommech}, 
with probability $1-0.1 \sigma$, 
$|d^* - d(S^*)| \leq O(\frac{1}{\epsilon}\log \frac{1}{\sigma})$.
Theorem~\ref{thm:util} now follows from 
this fact, as well as Claims~\ref{clm:bounderr}
and \ref{clm:approxanalysis}, and the choice of $\thresh$.
\elaine{note this affects choice of thresh}
\end{proof}

\section{Deferred Proofs for the Linear-Time Algorithm}
\label{sec:linearproof}

In this section, we prove Theorem~\ref{thm:linear}.
The DP proof is the same as Theorem~\ref{thm:dp}.
For the utility analysis, we may assume that whenever $\PSum(u)$ is updated
from $\Sigma$ to $\Sigma'$ for any vertex $u$, 
it holds that $\Sigma' - \Sigma < \err/2$ where $\err$ is the discretization
parameter used in the bucketing --- 
using the same type of arguments as the proof of Claim~\ref{clm:bounderr},
one can show that this indeed holds except with $0.1\sigma$ probability.
This means that no vertex $v$'s noisy residual degree $D(v) - \PSum(v)$ should
shrink by more than $\err/2 + 1$ in two adjacent time steps, i.e., 
whenever a vertex $v$ moves buckets, it cannot move left by more than 1 bucket.
Therefore, effectively, the only difference between this linear-time variant
and the earlier algorithm in Section~\ref{sec:mainalg}
is the following:
in this new variant, we do not necessarily pick
${\rm arg}\min_v(v, S_t)$  in every time step $t$. 
We could pick a vertex $v_t$
in time step $t$ such that $|v_t - {\rm arg}\min_v(v, S_t)| \leq \err$.
This introduces an additive $\err$ term in the proof
of Theorem~\ref{thm:util}.  
Due to the choice of $\err$, and redoing 
Theorem~\ref{thm:util} with the extra additive $\err$ term, 
we conclude that the above algorithm achieves 
$(2, O(\frac{1}{\epsilon} \cdot \log^{2.5}n \cdot \log\frac1\sigma))$-approximation
with probability $1-\sigma$.

We now bound the algorithm's runtime. 
First, not counting the runtime associated with \PSum updates, the rest
of the algorithm is easily seen to take only $O(n + m)$ time --- specifically,
the total work spent in Line 1
\elaine{hard coded ref} 
is at most $O(n)$ due to the same reason as Charika's non-DP, linear-time algorithm.
Due to the runtime of \PSum stated in Theorem~\ref{thm:psum}, 
to prove the statement about the runtime, 
it suffices to show that 
the number of $\PSum$ updates is upper 
bounded by $4m$ with probability $\exp(-\Omega(m))$.
Consider running the algorithm till it makes exactly $4m$ \PSum updates ---
if the algorithm ends before making $4m$ \PSum 
updates, we can simply pad it to 
exactly $4m$ number of 
$\PSum$ updates by appending filler \PSum updates at the end of the algorithm
which does not affect the outcome.
In this way, there is exactly one noise $\noise$ associated 
with each of the $4m$ \PSum updates.
Due to the choice of $\thresh$,
the probability that each $\noise$ is greater than $\thresh$ is at most $\sigma/n$.
Due to the Chernoff bound, 
the probability that there exist 
$3m$ or more noises $\noise$ that are greater  
than $\thresh$ is $\exp(-\Omega(m))$.
This means that for $m$ of these \PSum updates, the true increment 
input to the \PSum instance must be at least $1$, and thus one edge
must be consumed for each of these $m$ \PSum updates.
In other words, except with 
$\exp(-\Omega(m))$ probability, the algorithm must have ended
after having made $4m$ or fewer \PSum updates.

\elaine{now comment on correctness}

\section{Lower Bound}
\label{sec:lb}

\vspace{0.2cm}
{\noindent \textbf{Theorem \ref{thm:lb}}. \textit{Let $\alpha > 1$, $\epsilon > 0$ be arbitrary constants,
$\exp(-n^{0.49}) < \sigma < 0.000001 \cdot \min(1, \epsilon, \exp(-\epsilon))$, and $0 \le \delta \le \frac{\sigma \epsilon}{\log \frac{1}{4\sigma}}$.
Then, there exists a sufficiently small 
$\beta = \Theta\big(\frac{1}{\alpha}\sqrt{\frac{1}{\epsilon} \log \frac{1}{\sigma}}\big)$ 
such that 
there does not exist an $(\epsilon,\delta)$-DP mechanism that 
achieves $(\alpha, \beta)$-approximation 
with $1-\sigma$ probability.}}

\begin{proof}
Suppose there exists an $(\epsilon,\delta)$-DP mechanism 
denoted $M$ 
that achieves $(\alpha, \beta)$-approximation with $1-\sigma$ probability,
for the parameters stated above. Let $\beta = \frac{1}{100\alpha}\sqrt{\frac{1}{\epsilon}\log \frac{1}{4\sigma}}$. 
We will reach a contradition below. Thus, no $(\epsilon,\delta)$-DP mechanism can achieve a $(\alpha, \beta)$-approximation.

Consider a graph over $n$ vertices $V := [n]$. 
Suppose that a subset $A \subseteq V$ of size $4\alpha \beta +1$ 
of the vertices form a clique, and there are no other edges in $G$.
Therefore, the densest subgraph of $G$ is $A$, and its
density is $\frac{|A|-1}{2}$ which is $2\alpha \beta$. Note that for our choices of parameters, we always have $\alpha \beta \ge 1$.

If we run the mechanism $M$ over this graph $G$, we know that with probability 
$1-\sigma$, the true density of set of vertices output is at least $\beta$.
This means that with probability $1-\sigma$, 
at most $(4\alpha\beta)(4 \alpha \beta+1)/(2\beta) \le 16\alpha^2 \beta$ vertices are output, since the graph $G$ has $(4\alpha\beta)(4 \alpha \beta+1)/2$ total number of edges.
In other words, the expected number of vertices output  
is upper bounded by
\begin{equation}
(1-\sigma) \cdot 16\alpha^2 \beta + \sigma \cdot n \,.
\label{eqn:expectedv}
\end{equation}

We claim that there must exist a set $B$ of size $4\alpha\beta+1$ that is disjoint
from $A$ such that with probability at least $1/2$, \elaine{param}
no vertex in $B$ is output. 
Suppose that this is not the case, 
then, the expected number of 
vertices contained in the output 
is at least $s := \frac{n - (4\alpha \beta+1)}{4\alpha\beta+1}  
\cdot \frac12$.
For suitable choice of parameters as stated 
in the theorem statement, 
$s$ would be greater than 
Equation~(\ref{eqn:expectedv}) for sufficiently large $n$, which leads to a contradiction.
\elaine{this imposes some constraint on choice of delta}

Now, consider another graph $G'$ in which the set
$B$ of vertices form a clique and there are no other edges in $G'$. 
$G'$ can be obtained by making $2 \cdot (4 \alpha \beta+1)(4 \alpha \beta)/2 \le 32 \alpha^2 \beta^2$
 edge modifications
starting from $G$.

We now use the group privacy theorem to derive a lower bound on the probability that mechanism $M$ does not output any of the vertices in $B$ if we run it on $G'$.

\begin{theorem}[Group Privacy]
Let $M$ be a $(\epsilon, \delta)$-DP mechanism, and $G,G'$ be two datasets. Then, for any subset $U$ of the output space,
\[
\Pr[\Alg(G) \in U] \leq 
e^{k \epsilon} \cdot \Pr[\Alg(G') \in U]  + \delta \cdot k \cdot e^{k \epsilon} \,,
\]
where $k$ is the hamming distance between $G$ and $G'$.
\end{theorem}
By $(\epsilon,\delta)$-DP 
and the group privacy theorem, 
it holds that if we run the mechanism $M$ on the graph $G'$,
the probability that
none of $B$ is selected is at least
\begin{equation}
\begin{split}
\bigg(\frac{1}{2} - 32\alpha^2 \beta^2 \delta \exp(32\alpha^2 \beta^2 \epsilon)\bigg)   \cdot \exp(-\epsilon \cdot 32 \alpha^2 \beta^2) &= \frac{1}{2} \exp(-32\epsilon \cdot \alpha^2 \beta^2) - 32\alpha^2 \beta^2 \delta \nonumber \\
& > 2 \sigma - \sigma = \sigma
 \,,
 \end{split}
\end{equation}
where the inequality above is because $\beta \le \frac{1}{100\alpha}\sqrt{\frac{1}{\epsilon}\log \frac{1}{4\sigma}}$, and $\delta \le \frac{\sigma\epsilon}{\log \frac{1}{4\sigma}}$.
Thus, with larger than $\sigma$ probability, the true density of the set
of vertices output is $0$.
Therefore, it is impossible that the mechanism $M$ gives
$(\alpha, \beta)$-approximation with $1-\sigma$ probability
(over any graph).  
\end{proof}

\bibliographystyle{apalike} 
\newpage



\end{document}